\documentclass[12pt]{amsart} 

\usepackage[margin=1in]{geometry}

\usepackage{amssymb,amsfonts,amsmath,latexsym,amsthm}
\usepackage[usenames,dvipsnames]{color}
\usepackage[]{graphicx}
\usepackage[space]{grffile}
\usepackage{mathrsfs}   
\usepackage{caption}
\usepackage{subcaption}
\usepackage{verbatim}
\usepackage{url}
\usepackage[colorlinks,citecolor=blue,linkcolor=blue,urlcolor=blue]{hyperref}
\usepackage[authoryear,round]{natbib}

\theoremstyle{plain}
\newtheorem{theorem}{Theorem}[section]
\newtheorem{proposition}[theorem]{Proposition}

\theoremstyle{definition}
\numberwithin{equation}{section}

\newcommand{\DPM}{{\scalebox{0.5}{\mbox{DPM}}}}
\newcommand{\Wishart}{\mathrm{Wishart}}
\newcommand{\Ga}{\mathrm{Gamma}}

\newcommand{\T}{\mathtt{T}}
\newcommand{\R}{\mathbb{R}}
\newcommand{\Z}{\mathbb{Z}}

\newcommand{\C}{\mathcal{C}}
\newcommand{\M}{\mathcal{M}}
\newcommand{\N}{\mathcal{N}}
\newcommand{\X}{\mathcal{X}}

\renewcommand{\emptyset}{\varnothing}
\newcommand{\iid}{\stackrel{\mathrm{iid}}{\sim}}

\newcommand{\matrixsmall}[1]{\bigl(\begin{smallmatrix}#1\end{smallmatrix} \bigr)}

\newcommand{\branch}[4]{
\left\{
	\begin{array}{ll}
		#1  & \mbox{if } #2 \\
		#3 & \mbox{if } #4
	\end{array}
\right.
}

\title{Mixture models with a prior on the number of components}
\author{Jeffrey W. Miller and Matthew T. Harrison}

\begin{document}
\begin{abstract}
A natural Bayesian approach for mixture models with an unknown number of components is to take the usual finite mixture model with Dirichlet weights, and put a prior on the number of components---that is, to use a mixture of finite mixtures (MFM).  While inference in MFMs can be done with methods such as reversible jump Markov chain Monte Carlo, it is much more common to use Dirichlet process mixture (DPM) models because of the relative ease and generality with which DPM samplers can be applied.  In this paper, we show that, in fact, many of the attractive mathematical properties of DPMs are also exhibited by MFMs---a simple exchangeable partition distribution, restaurant process, random measure representation, and in certain cases, a stick-breaking representation. Consequently, the powerful methods developed for inference in DPMs can be directly applied to MFMs as well.  We illustrate with simulated and real data, including high-dimensional gene expression data.  
\end{abstract}

\maketitle

\section{Introduction}
\label{section:intro}

Mixture models are used in a wide range of applications, including
population structure \citep{Pritchard_2000}, 
document modeling \citep{Blei_2003b},
speaker recognition \citep{Reynolds_2000},
computer vision \citep{stauffer1999adaptive},
phylogenetics \citep{pagel2004phylogenetic}, and
gene expression profiling \citep{Yeung_2001}, to name a few prominent examples.
A common issue with finite mixtures is that it can be difficult to choose an appropriate number of mixture components, and many methods have been proposed for making this choice \citep[e.g., ][]{Henna_1985,Keribin_2000,Leroux_1992,Ishwaran_2001,James_2001}. 

From a Bayesian perspective, perhaps the most natural approach is to treat the number of components like any other unknown parameter and put a prior on it.  For short, we refer to such a model as a mixture of finite mixtures (MFM).
Several inference methods have been proposed for this type of model \citep{Nobile_1994,Phillips_1996,Richardson_1997,Stephens_2000,Nobile_2007}, the most commonly-used method being reversible jump Markov chain Monte Carlo \citep{green1995reversible,Richardson_1997}.  Reversible jump is a very general technique, and has been successfully applied in many contexts, but it is perceived to be difficult to use, and applying it to new situations requires one to design good reversible jump moves, which can be nontrivial, particularly in high-dimensional parameter spaces.

Meanwhile, infinite mixture models such as Dirichlet process mixtures (DPMs) have become popular, partly due to the existence of generic Markov chain Monte Carlo (MCMC) algorithms that can easily be adapted to new applications \citep{Neal_1992,Neal_2000,MacEachern_1994,MacEachern_1998b,MacEachern_1998,Bush_1996,West_1992,West_1994,Escobar_1995,Liu_1994,Dahl_2003,Dahl_2005,Jain_2004,Jain_2007}.
These algorithms are made possible by the fact that the Dirichlet process has a variety of elegant mathematical properties---an exchangeable partition distribution, the Blackwell--MacQueen urn process (a.k.a.\ the Chinese restaurant process), a random discrete measure formulation, and the Sethuraman--Tiwari stick-breaking representation \citep{Ferguson_1973,Antoniak_1974,Blackwell_1973,aldous1985exchangeability,Pitman_1995,Pitman_1996,Sethuraman_1994,Sethuraman_1981}.

The purpose of this paper is to show that in fact, MFMs typically exhibit many of these same appealing properties---an exchangeable partition distribution, urn/restaurant process, random discrete measure formulation, and in certain cases, a simple stick-breaking representation---and consequently, that many of the inference techniques developed for DPMs can be directly applied to MFMs.
In particular, these properties enable one to do inference in MFMs without using reversible jump.
Interestingly, the key properties of MFMs hold for any choice of prior distribution on the number of components.

There has been a large amount of research on efficient inference methods for DPMs, and an immediate consequence of the present work is that most of these methods can also be used for MFMs.  Since many DPM sampling algorithms (for both conjugate and non-conjugate priors) are designed to have good mixing properties across a wide range of applications---for instance, the Jain--Neal split-merge samplers \citep{Jain_2004,Jain_2007}, coupled with incremental Gibbs moves \citep{MacEachern_1994,Neal_1992,Neal_2000}---this greatly simplifies the use of MFMs in new applications.


This work resolves an open problem discussed by \citet{Green_2001}, who noted that it would be interesting to be able to apply DPM samplers to MFMs: 
\begin{quote}
\emph{``In view of the intimate correspondence between DP and [MFM] models discussed above, it is 
interesting to examine the possibilities of using either class of MCMC methods for the other 
model class. We have been unsuccessful in our search for incremental Gibbs samplers for the 
[MFM] models, but it turns out to be reasonably straightforward to implement reversible jump split/merge methods for DP models.''}
\end{quote}

The paper is organized as follows. In the remainder of this section, we motivate this work with an overview of the similarities and differences between MFMs and DPMs, illustrated by a simulation example. In Sections \ref{section:model} and \ref{section:partitions}, we formally define the MFM and show that it gives rise to a simple exchangeable partition distribution closely paralleling that of the Dirichlet process. In Section \ref{section:representations}, we present the P\'olya urn scheme (restaurant process), random discrete measure formulation, and stick-breaking representation for the MFM.  In Section \ref{section:asymptotics}, we establish some asymptotic results for MFMs.  In Section \ref{section:inference}, we show how the properties in Sections \ref{section:partitions} and \ref{section:representations} lead to efficient inference algorithms for the MFM, and in Section \ref{section:empirical}, we apply the model to the galaxy dataset (a standard benchmark) and to high-dimensional gene expression data used to discriminate cancer subtypes.  We close with a brief discussion.

\subsection{Background: Similarities and differences between MFMs and DPMs}
\label{section:comparison}

In many respects, MFMs and DPMs are quite similar, but there are some important differences.
In general, we would not say that one model is uniformly better than the other; rather, one should choose the model which is best suited to the application at hand.

\subsubsection*{Density estimation}
\begin{figure}
\centering
\includegraphics[trim=1cm 0.5cm 1cm 0.2cm, clip=true, width=0.49\textwidth]{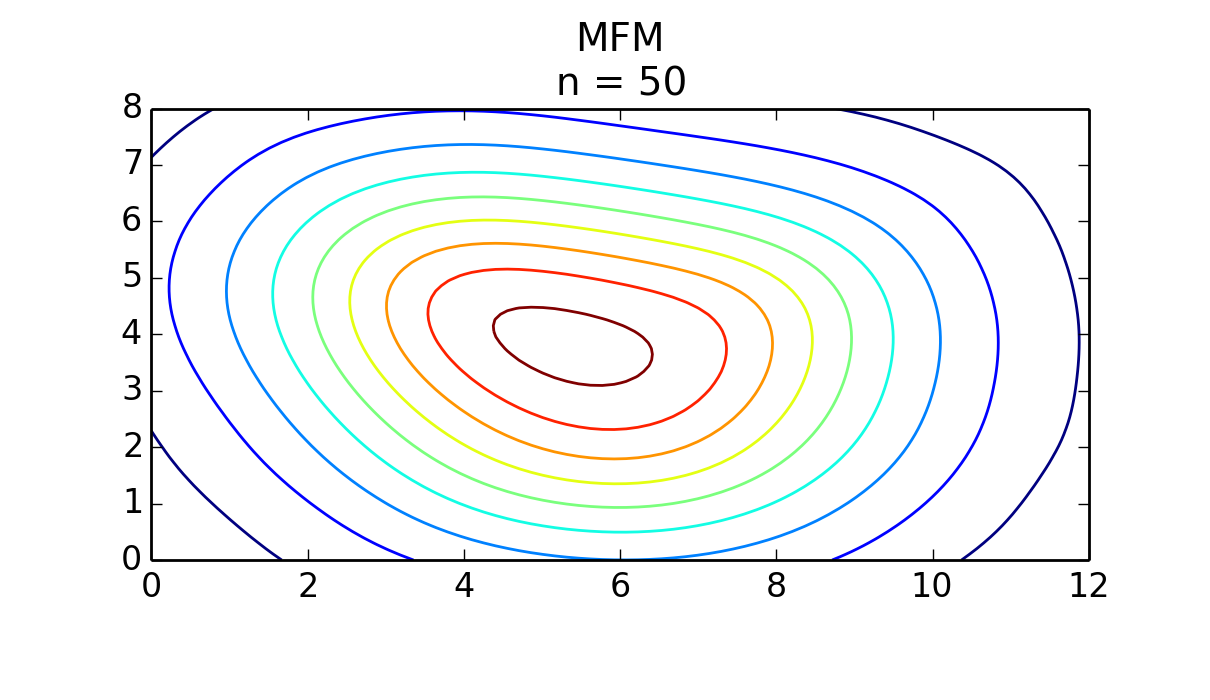}
\includegraphics[trim=1cm 0.5cm 1cm 0.2cm, clip=true, width=0.49\textwidth]{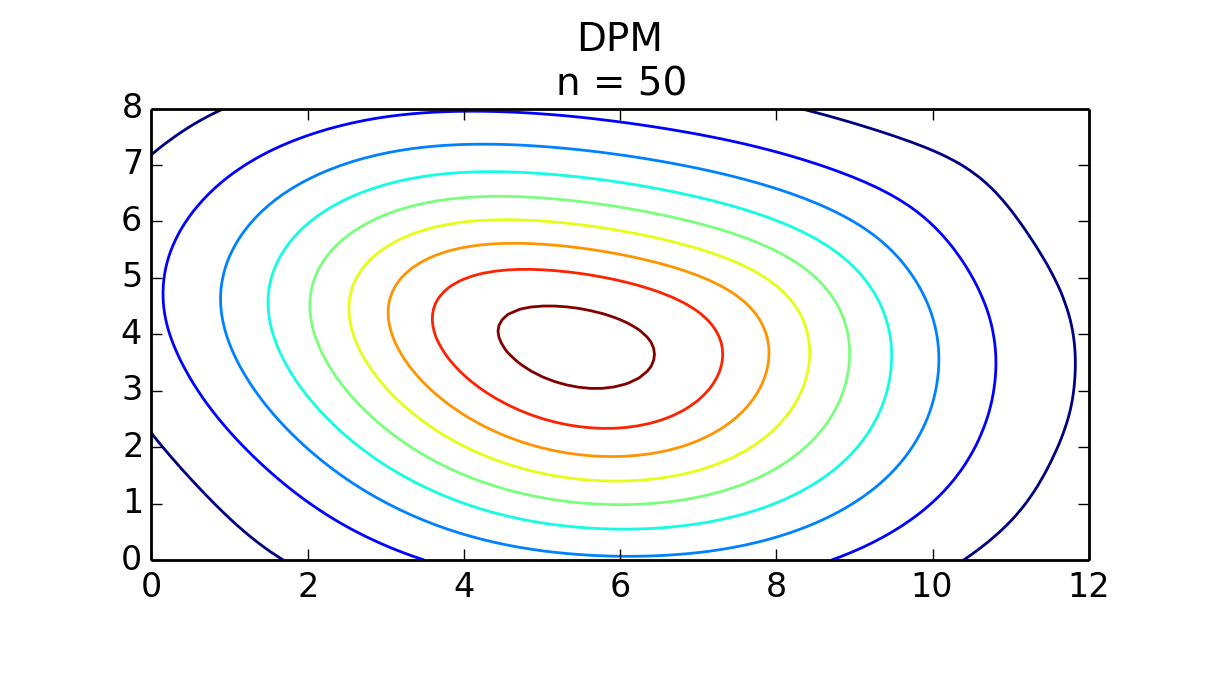}
\includegraphics[trim=1cm 0.5cm 1cm 0.2cm, clip=true, width=0.49\textwidth]{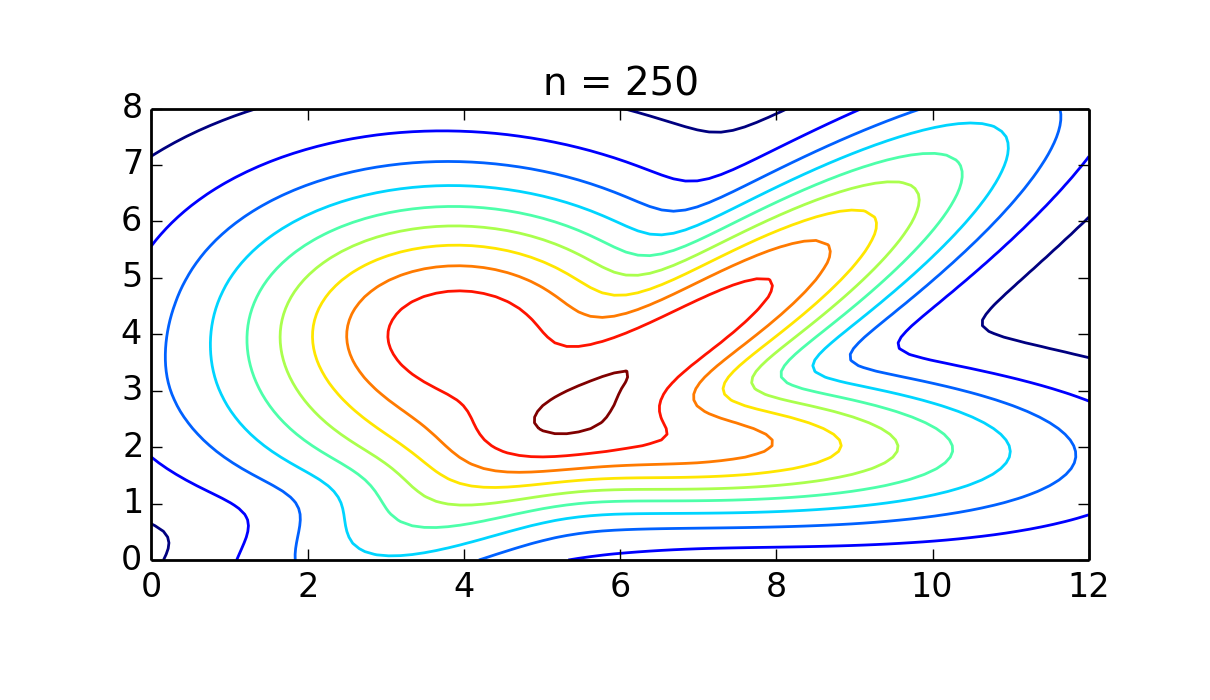}
\includegraphics[trim=1cm 0.5cm 1cm 0.2cm, clip=true, width=0.49\textwidth]{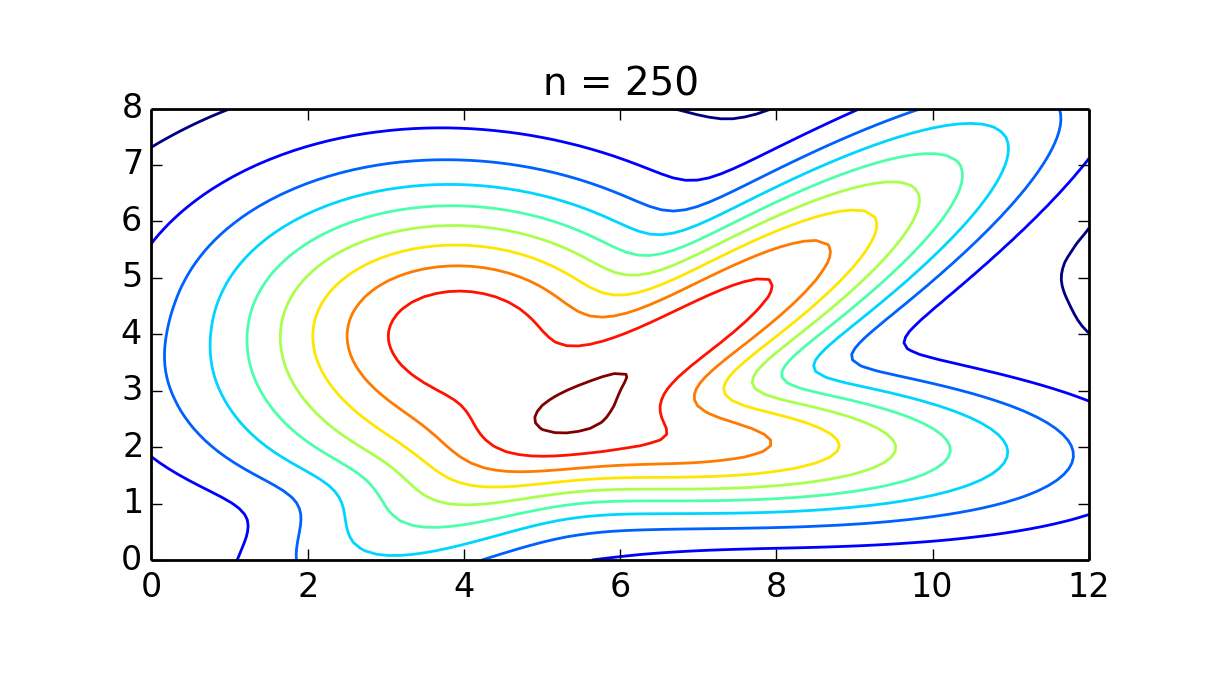}
\includegraphics[trim=1cm 0.5cm 1cm 0.2cm, clip=true, width=0.49\textwidth]{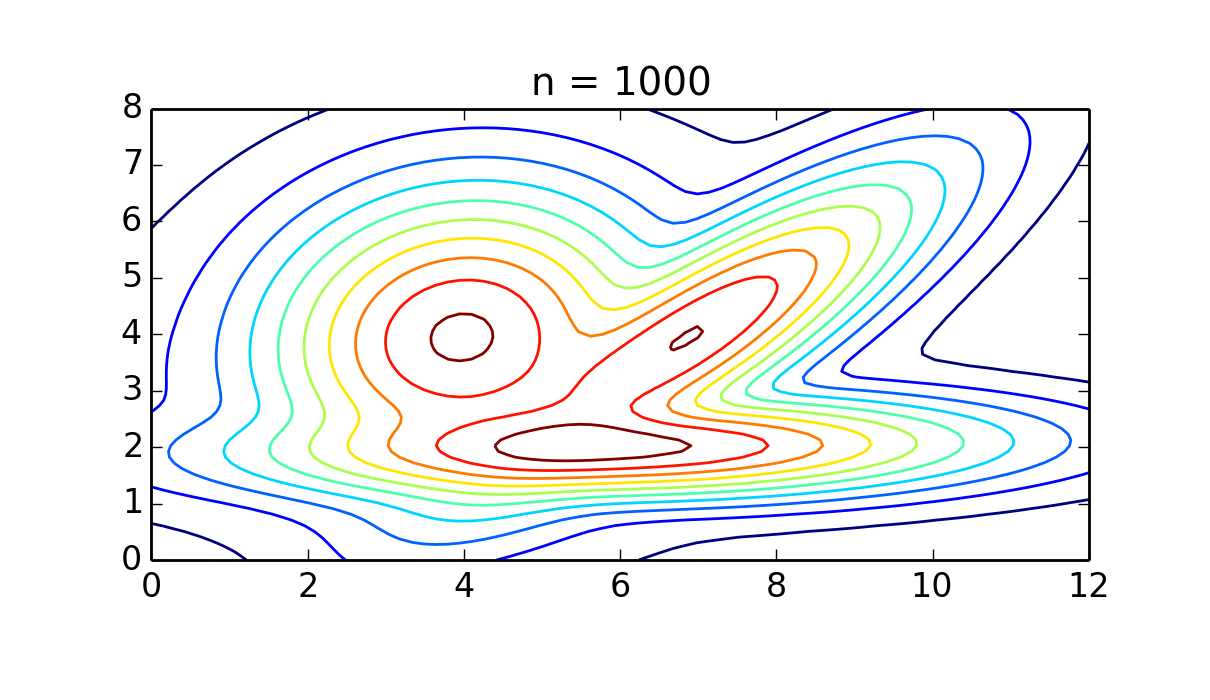}
\includegraphics[trim=1cm 0.5cm 1cm 0.2cm, clip=true, width=0.49\textwidth]{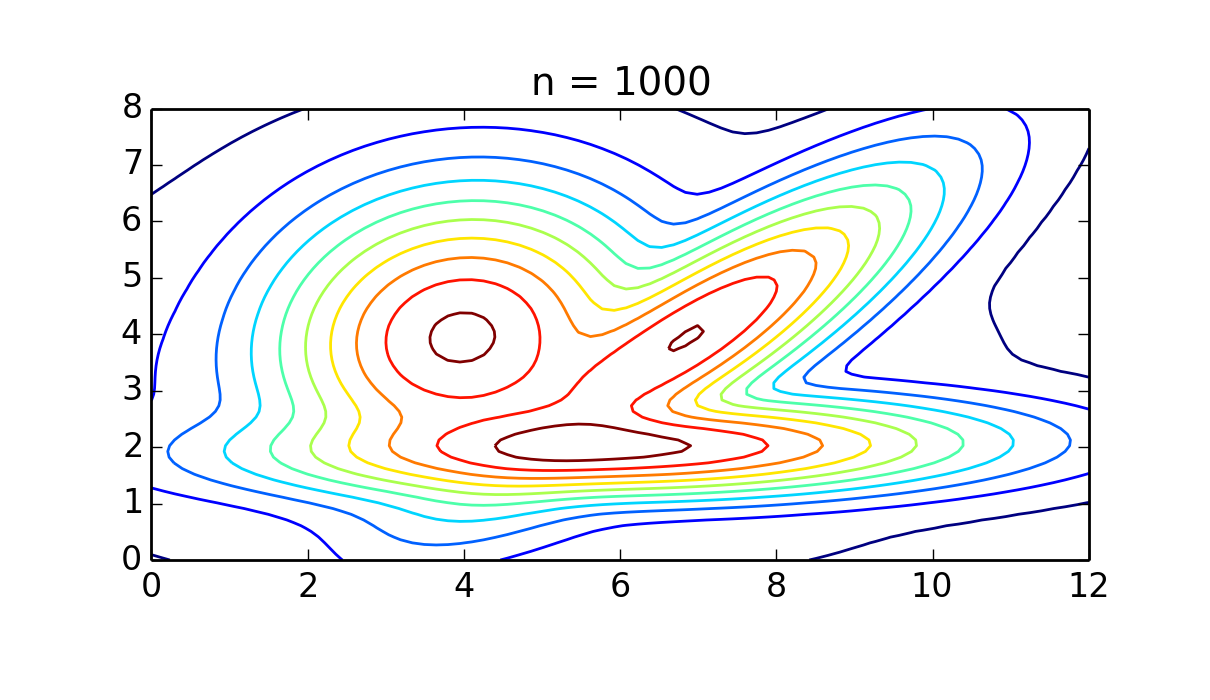}
\includegraphics[trim=1cm 0.5cm 1cm 0.2cm, clip=true, width=0.49\textwidth]{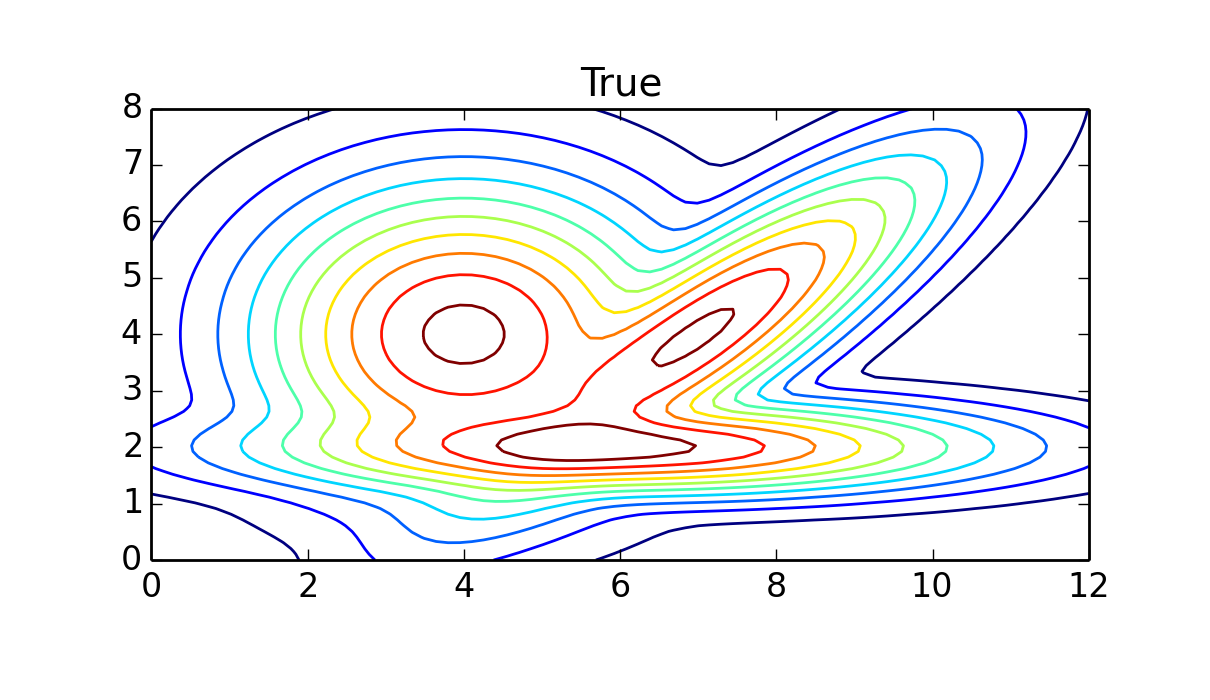}
\caption{Density estimates for MFM (left) and DPM (right) on increasing amounts of data from a three-component Gaussian mixture (bottom). As $n$ increases, the estimates appear to be converging to the true density, as expected. 
See Section \ref{section:simulation} for details.}
\label{figure:density-estimates}
\end{figure}

\begin{figure}
\centering
\includegraphics[trim=0cm 0.5cm 0cm 0.2cm, clip=true, width=0.75\textwidth]{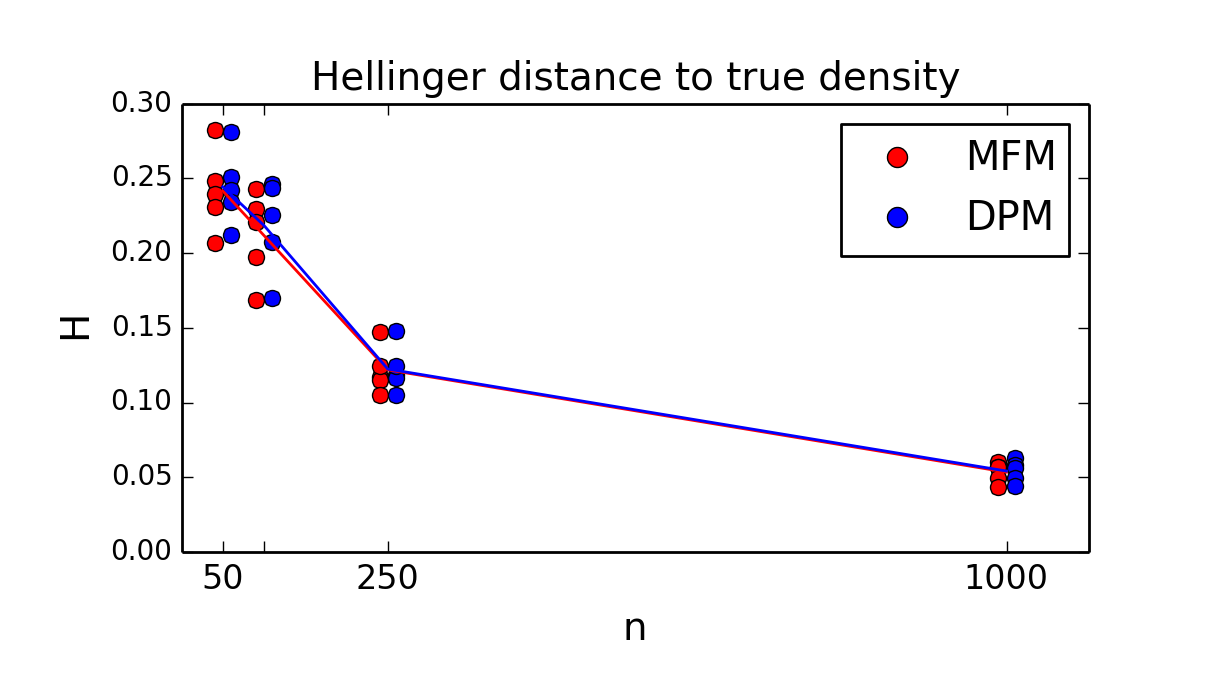}
\caption{Hellinger distance to the true density, for MFM (red, left) and DPM (blue, right) density estimates, on data from a three-component Gaussian mixture. For each $n\in\{50,100,250,1000\}$, five independent datasets of size $n$ were used, and the lines connect the averages of the distances for each $n$.
See Section \ref{section:simulation} for details.}
\label{figure:Hellinger}
\end{figure}

For certain nonparametric density estimation problems, both models have been shown to exhibit posterior consistency at the minimax optimal rate, up to logarithmic factors \citep{Kruijer_2010,Ghosal_2007}.
Even for small sample sizes, we observe empirically that density estimates under the two models are remarkably similar.  As a toy example, Figure \ref{figure:density-estimates} compares their density estimates on data from a bivariate Gaussian mixture with three components.  As the amount of data increases, these density estimates appear to be converging to the true density, as expected; indeed, Figure \ref{figure:Hellinger} indicates that the Hellinger distance to the true density is going to zero.

\begin{figure}
\centering
\includegraphics[trim=1cm 0.5cm 1cm 0.2cm, clip=true, width=0.49\textwidth]{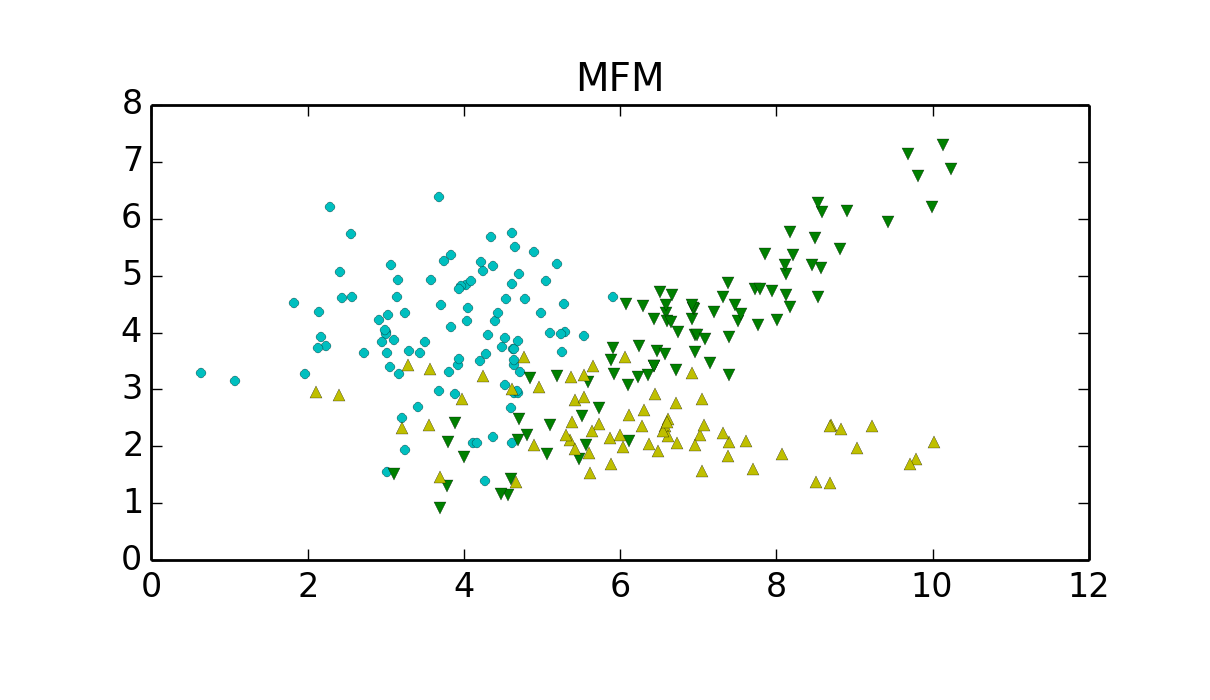}
\includegraphics[trim=1cm 0.5cm 1cm 0.2cm, clip=true, width=0.49\textwidth]{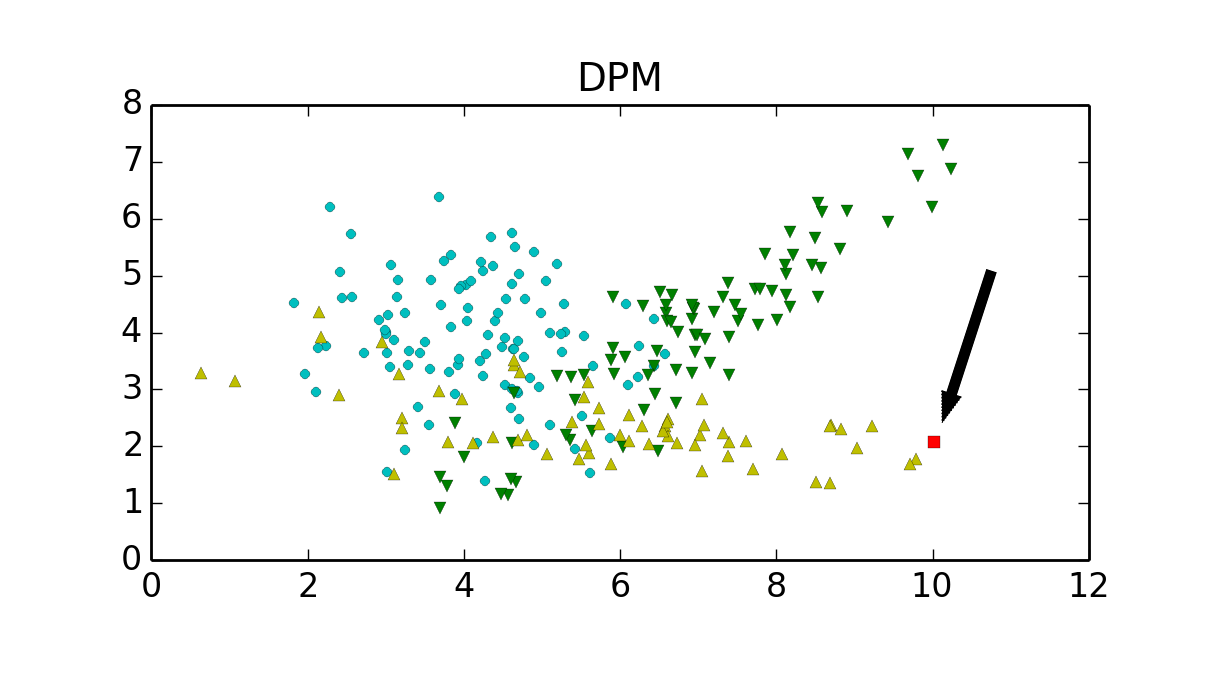}
\includegraphics[trim=1cm 0.5cm 1cm 0.2cm, clip=true, width=0.49\textwidth]{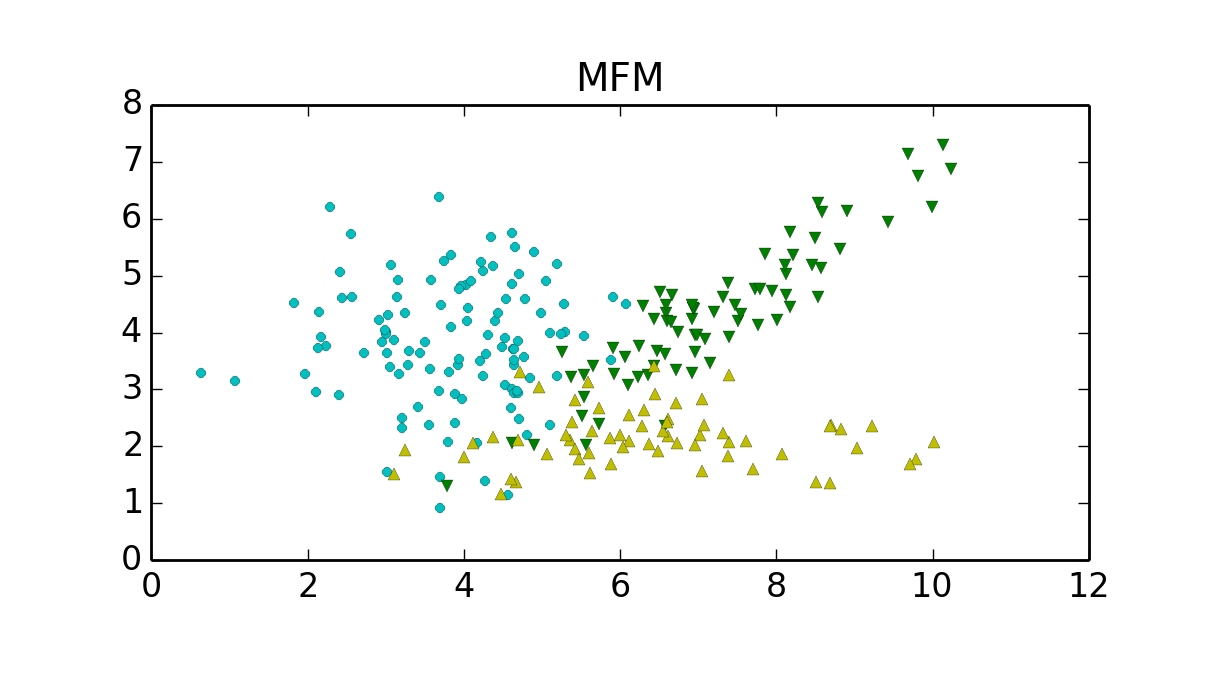}
\includegraphics[trim=1cm 0.5cm 1cm 0.2cm, clip=true, width=0.49\textwidth]{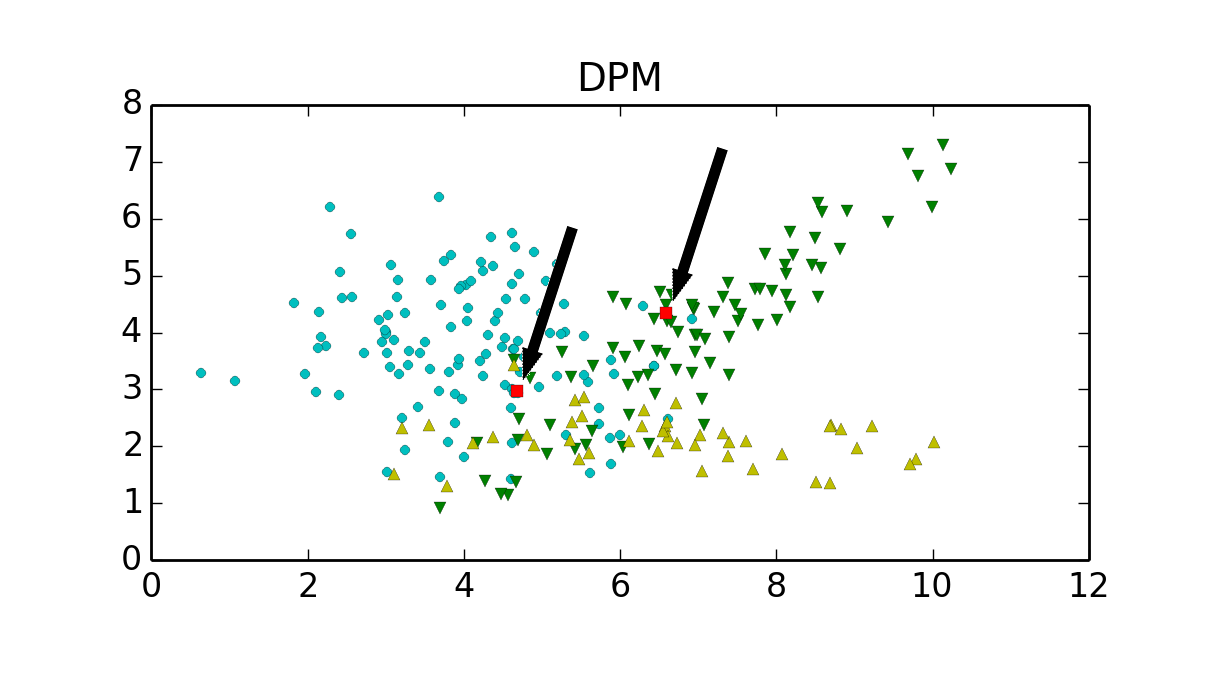}
\includegraphics[trim=1cm 0.5cm 1cm 0.2cm, clip=true, width=0.49\textwidth]{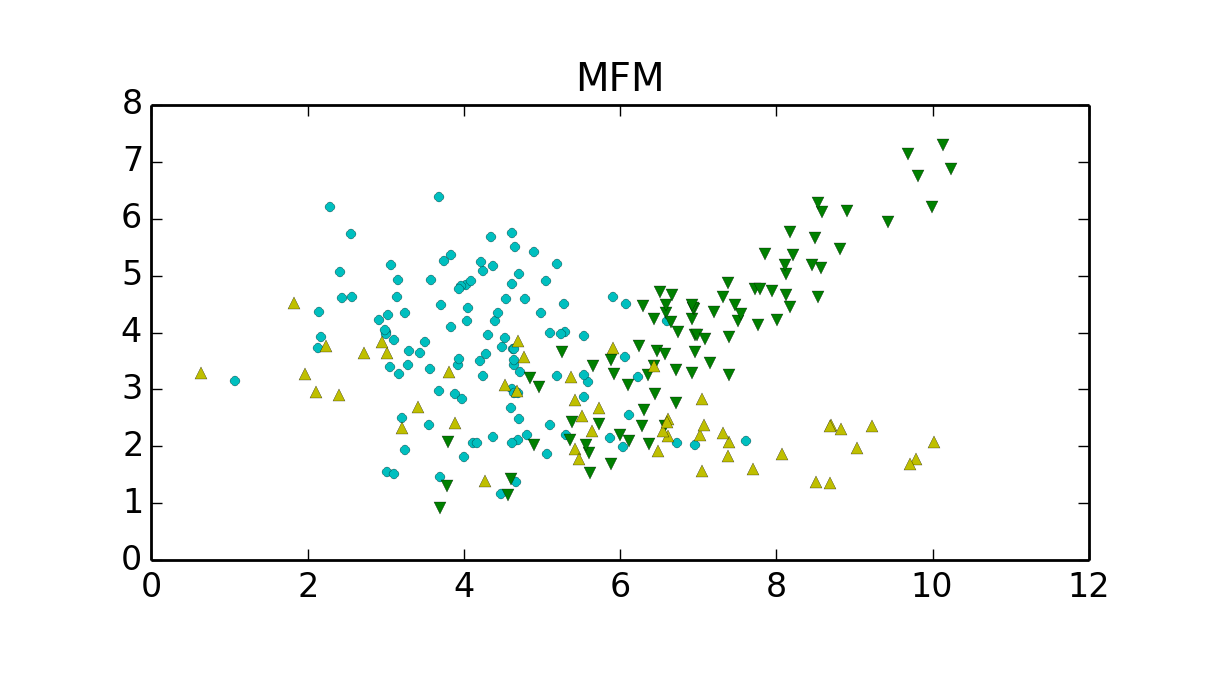}
\includegraphics[trim=1cm 0.5cm 1cm 0.2cm, clip=true, width=0.49\textwidth]{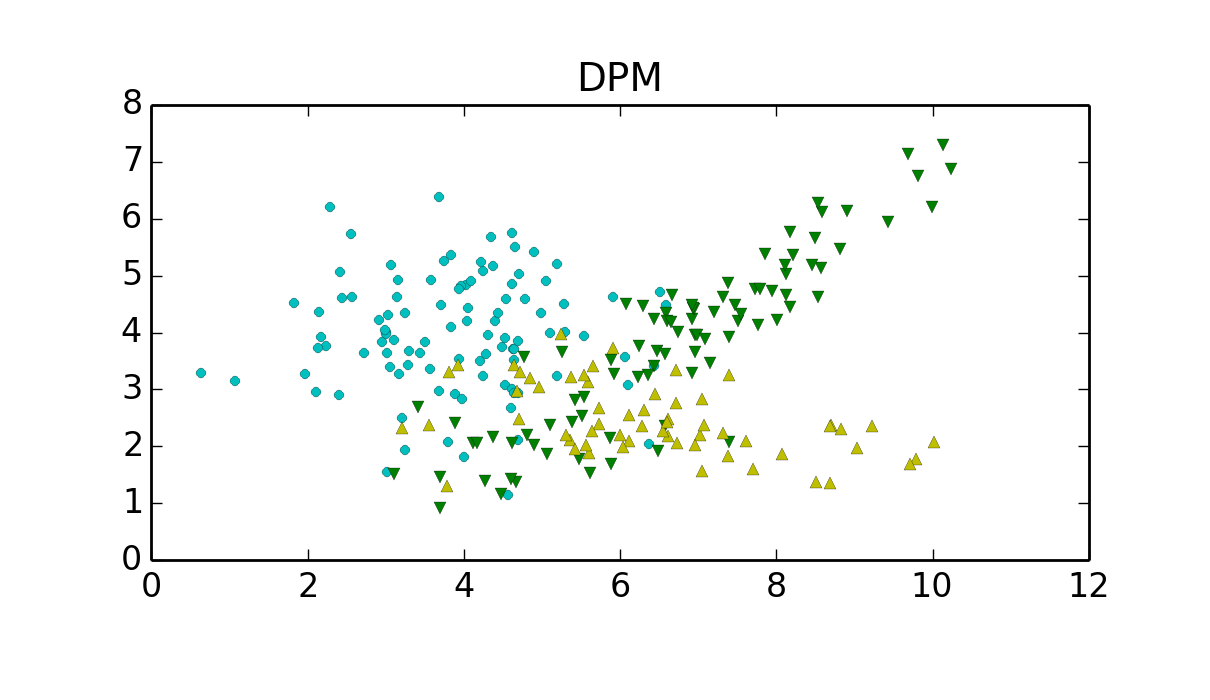}
\includegraphics[trim=1cm 0.5cm 1cm 0.2cm, clip=true, width=0.49\textwidth]{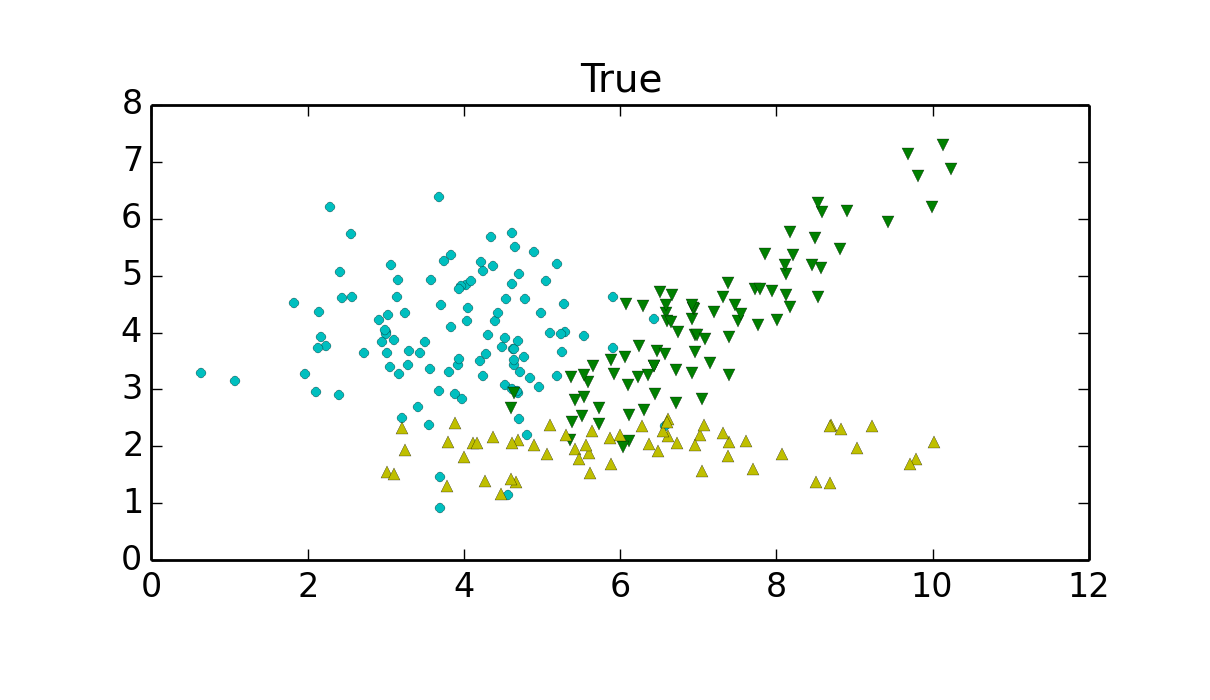}
\caption{Typical sample clusterings from the posterior for the MFM (left) and DPM (right), on $n = 250$ data points from a three-component Gaussian mixture; the bottom plot shows the true component assignments. Note the small extra clusters in the DPM samples (red squares).  Best viewed in color.
See Section \ref{section:simulation} for details.}
\label{figure:clustering}
\end{figure}

\subsubsection*{Clustering}

It seems that more often, mixture models are used for clustering and latent class discovery rather than density estimation. MFMs have a partition distribution that takes a very similar form to that of the Dirichlet process; see Section \ref{section:partitions}.  However, despite this similarity, the MFM partition distribution differs in two fundamental respects.  While the first is widely known, the second is far less often appreciated, and yet is perhaps even more important.

\begin{enumerate}
\item[(1)] The prior on the number of clusters $t$ is very different.  In an MFM, one has complete control over the prior on the number of components $k$, which in turn provides control over the prior on $t$. As the sample size $n$ grows, in an MFM the prior on $t$ converges to the prior on $k$ (in fact, $t$ converges to $k$ almost surely). In contrast, in a Dirichlet process, the prior on $t$ takes a particular parametric form and diverges at a $\log n$ rate.
\item[(2)] Given $t$, the prior on the cluster sizes is very different. In an MFM, most of the prior mass is on partitions in which the sizes of the clusters are all the same order of magnitude, while in a Dirichlet process, most of the prior mass is on partitions in which the sizes vary widely, with a few large clusters and many very small clusters.
\end{enumerate}

See Section \ref{section:asymptotics} for more precise descriptions of (1) and (2) in mathematical terms. (Note that while some authors use the terms ``cluster'' and ``component'' interchangeably, we use \emph{cluster} to refer to a group of data points, and \emph{component} to refer to one of the probability distributions in a mixture model.)

These prior differences carry over to noticeably different posterior clustering behavior.
For instance, Figure \ref{figure:clustering} displays typical clusterings sampled from the posterior, illustrating that DPM samples tend to have tiny ``extra'' clusters, while MFM samples do not.

It should also be mentioned that due to (2), MFMs ``dislike'' partitions with very small clusters, causing incremental Gibbs samplers to mix more slowly (empirically) when $n$ is large, however, this is easily remedied by using split-merge samplers \citep{Jain_2004,Jain_2007}; see Section \ref{section:inference}.

\subsubsection*{Mixing distribution and the number of components}

Assuming that the data is from a finite mixture, it is also sometimes of interest to infer the mixing distribution or the number of components, subject to the caveat that these inferences are meaningful only to the extent that the component distributions are correctly specified and the model is mixture identifiable.  While \citet{Nguyen_2013} has shown that under certain conditions, DPMs are consistent for the mixing distribution (in the Wasserstein metric), \citet{Miller_2014} have shown that the posterior on the number of clusters in a DPM is typically not consistent for the number of components. On the other hand, MFMs are consistent for the mixing distribution and the number of components (for Lebesgue almost-all parameter values) under very general conditions; this is a straightforward consequence of Doob's theorem \citep{Nobile_1994}.  The relative ease with which this consistency can be established for MFMs is due to the fact that in an MFM, the parameter space is a countable union of finite-dimensional spaces, rather than an infinite-dimensional space.

\begin{figure}
\centering
\includegraphics[trim=0cm 0cm 0cm 0cm, clip=true, width=0.49\textwidth]{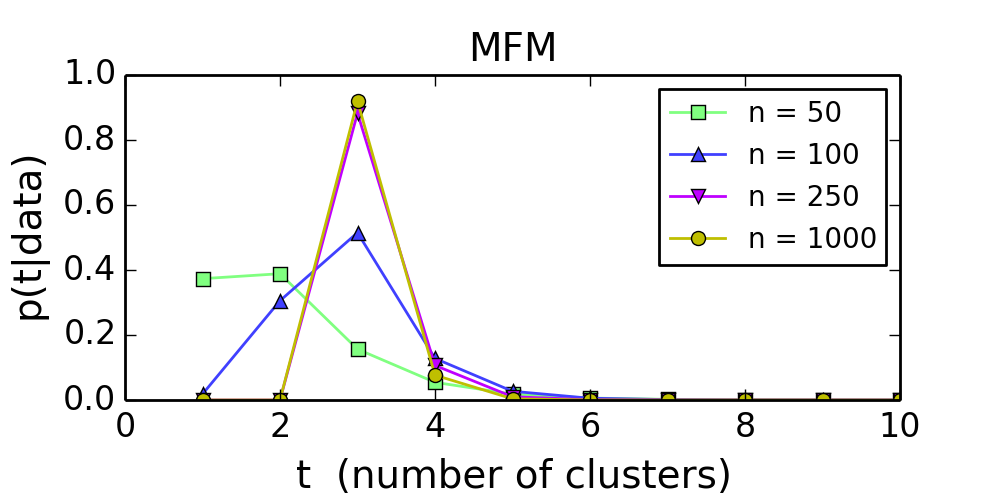}
\includegraphics[trim=0cm 0cm 0cm 0cm, clip=true, width=0.49\textwidth]{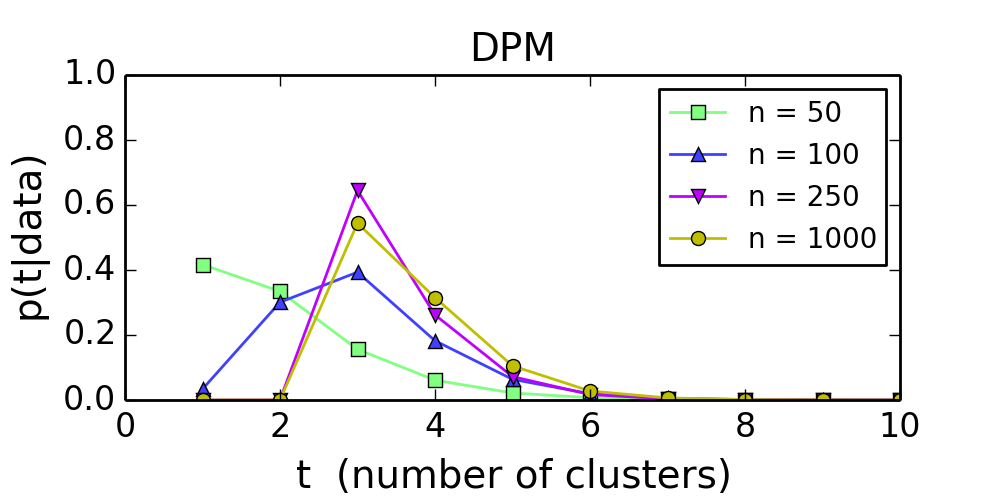} \\
\includegraphics[trim=0cm 0cm 0cm 0cm, clip=true, width=0.49\textwidth]{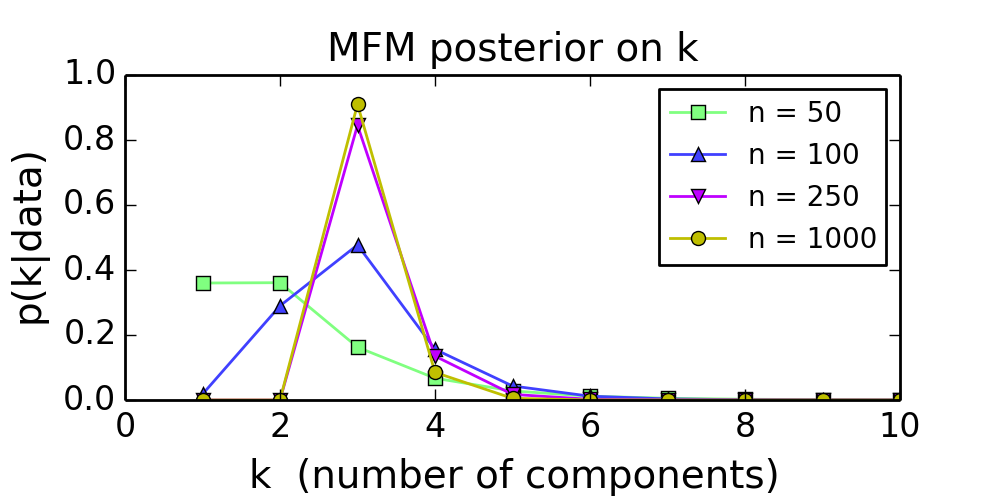}
\caption{Posterior on the number of clusters $t$ for the MFM (top left) and DPM (top right), and the posterior on the number of components $k$ for the MFM (bottom), on increasing amounts of data from a three-component Gaussian mixture.
See Section \ref{section:simulation} for details.}
\label{figure:tk-posteriors}
\end{figure}

These consistency/inconsistency properties are readily observed empirically---they are not simply large-sample phenomena. As seen in Figure \ref{figure:tk-posteriors}, the tendency of DPM samples to have tiny extra clusters causes the number of clusters $t$ to be somewhat inflated, apparently making the DPM posterior on $t$ fail to concentrate, while the MFM posterior on $t$ concentrates at the true value (see Section \ref{section:clusters-components-relationship}).  In addition to the number of clusters $t$, the MFM also permits inference for the number of components $k$ in a natural way (Figure \ref{figure:tk-posteriors}), while in the DPM the number of components is always infinite.  See Section \ref{section:discussion} for discussion regarding issues with estimating the number of components.


\section{Model}
\label{section:model}

We consider the following well-known model:
\begin{align}
& K\sim p_K, \mbox{ where $p_K$ is a p.m.f.\ on $\{1,2,\dotsc\}$}\notag\\
& (\pi_1,\dotsc,\pi_k)\sim \mathrm{Dirichlet}_k(\gamma,\dotsc,\gamma), \mbox{ given $K = k$}\notag\\
& Z_1,\dotsc,Z_n\iid \pi, \mbox{ given $\pi$}\label{equation:model}\\
& \theta_1,\dotsc,\theta_k\iid H, \mbox{ given $K = k$}\notag\\
& X_j\sim f_{\theta_{Z_j}} \mbox{ independently for $j = 1,\dotsc,n$, given $\theta_{1:K},Z_{1:n}$.}\notag
\end{align}
Here, $H$ is a prior or ``base measure'' on $\Theta\subset\R^\ell$, and $\{f_\theta:\theta\in\Theta\}$ is a family of probability densities with respect to a sigma-finite measure $\zeta$ on $\X\subset\R^d$. (As usual, we give $\Theta$ and $\X$ the Borel sigma-algebra, and assume $(x,\theta)\mapsto f_\theta(x)$ is measurable.) We denote $x_{1:n} = (x_1,\ldots,x_n)$.  Typically, the values $X_1,\dotsc,X_n$ would be observed, and all other variables would be hidden/latent. We refer to this as a \emph{mixture of finite mixtures} (MFM) model.

It is important to note that we assume a symmetric Dirichlet with a single parameter $\gamma$ not depending on $k$. This assumption is key to deriving a simple form for the partition distribution and the other resulting properties. 
Assuming symmetry in the distribution of $\pi$ is quite natural, since the distribution of $X_1,\dotsc,X_n$ under any asymmetric distribution on $\pi$ would be the same as if this were replaced by its symmetrized version, i.e., if the entries of $\pi$ were uniformly permuted (although this would no longer necessarily be a Dirichlet distribution). 
Assuming the same $\gamma$ for all $k$ is a genuine restriction, albeit a fairly natural one, often made in such models even when not strictly necessary \citep{Nobile_1994,Phillips_1996,Richardson_1997,Green_2001,Stephens_2000,Nobile_2007}.
Note that prior information about the relative sizes of the mixing weights $\pi_1,\dotsc,\pi_k$ can be introduced through $\gamma$---roughly speaking, small $\gamma$ favors lower entropy $\pi$'s, while large $\gamma$ favors higher entropy $\pi$'s.

Meanwhile, we put very few restrictions on $p_K$, the distribution of the number of components. For practical purposes, we need the infinite series $\sum_{k = 1}^\infty p_K(k)$ to converge to $1$ reasonably quickly, but any choice of $p_K$ arising in practice should not be a problem. For certain theoretical purposes---in particular, consistency for the number of components---it is desirable to have $p_K(k)>0$ for all $k\in\{1,2,\dotsc\}$.

For comparison, the Dirichlet process mixture (DPM) model with concentration parameter $\alpha>0$ and base measure $H$ is defined as follows, using Sethuraman's \citeyearpar{Sethuraman_1994} representation:
\begin{align*}
& B_1,B_2,\dotsc\iid\mathrm{Beta}(1,\alpha)\\
& Z_1,\dotsc,Z_n\iid \pi, \mbox{ given } \pi=(\pi_1,\pi_2,\dotsc) \mbox{ where } \pi_i = B_i\textstyle\prod_{j=1}^{i-1} (1 - B_j)\\
& \theta_1,\theta_2,\dotsc\iid H\\
& X_j\sim f_{\theta_{Z_j}} \mbox{ independently for $j = 1,\dotsc,n$, given $\theta_{1:\infty},Z_{1:n}$.}
\end{align*}

\section{Exchangeable partition distribution}
\label{section:partitions}

The primary observation on which our development relies is that the distribution on partitions induced by an MFM takes a form which is simple enough that it can be easily computed. 
Let $\C$ denote the unordered partition of $[n]:=\{1,\dotsc,n\}$ induced by $Z_1,\dotsc,Z_n$; in other words, $\C =\{E_i:|E_i|>0\}$ where $E_i =\{j: Z_j = i\}$ for $i\in\{1,2,\dotsc\}$.

\begin{theorem}
\label{theorem:EPPF}
Under the MFM (Equation \ref{equation:model}), the probability mass function of $\C$ is
\begin{align}\label{equation:EPPF}
p(\C) = V_n(t)\prod_{c\in\C} \gamma^{(|c|)}
\end{align}
where $t=|\C|$ is the number of parts in the partition, and
\begin{align}\label{equation:v}
 V_n(t) =\sum_{k = 1}^\infty\frac{k_{(t)}}{(\gamma k)^{(n)}}\,p_K(k).
\end{align}
\end{theorem}

All proofs have been collected in Appendix \ref{section:proofs}.  Here, $x^{(m)} =x(x+1)\cdots(x + m-1)$ and $x_{(m)} =x(x-1)\cdots(x - m+1)$, with $x^{(0)} = 1$ and $x_{(0)} = 1$ by convention.  We discuss computation of $V_n(t)$ in Section \ref{section:v}.  For comparison, under the DPM, the partition distribution induced by $Z_1,\ldots,Z_n$ is $p_\DPM(\C)=\frac{\alpha^t}{\alpha^{(n)}} \prod_{c\in\C} (|c|-1)!$ \citep{Antoniak_1974}.

Viewed as a function of the part sizes $(|c|:c\in\C)$, Equation \ref{equation:EPPF} is an \emph{exchangeable partition probability function} (EPPF) in the terminology of \cite{Pitman_1995,Pitman_2006}, since it is a symmetric function of the part sizes. Consequently, $\C$ is an \emph{exchangeable random partition} of $[n]$; that is, its distribution is invariant under permutations of $[n]$ (alternatively, this can be seen directly from the definition of the model, since $Z_1,\dotsc,Z_n$ are exchangeable). 

More specifically, we observe that Equation \ref{equation:EPPF} is a member of the family of Gibbs partition distributions \citep{Pitman_2006}; this is also implied by the results of \cite{Gnedin_2006} characterizing the extreme points of the space of Gibbs partition distributions. Further results on Gibbs partitions are provided by \citet{Ho_2007}, \citet{Lijoi_2008}, \citet{Cerquetti_2008,Cerquetti_2011}, \citet{Gnedin_2010}, and \citet{Lijoi_2010}. However, the utility of this representation for inference in mixture models with a prior on the number of components does not seem to have been previously explored in the literature.

Due to Theorem \ref{theorem:EPPF}, we have the following equivalent representation of the model:
\begin{align}
& \C\sim p(\C), \mbox{ with $p(\C)$ as in Equation \ref{equation:EPPF}}\notag\\
& \phi_c\iid H  \mbox{ for $c\in\C$, given $\C$}\label{equation:model-C}\\
& X_j\sim f_{\phi_c} \mbox{ independently for $j \in c$, $c\in\C$, given $\phi,\C$,}\notag
\end{align}
where $\phi =(\phi_c: c\in\C)$ is a tuple of $t =|\C|$ parameters $\phi_c\in\Theta$, one for each part $c\in\C$.

This representation is particularly useful for doing inference, since one does not have to deal with cluster labels or empty components. The formulation of models starting from a partition distribution has been a fruitful approach, exemplified by the development of product partition models \citep{Hartigan_1990,Barry_1992,Quintana_2003,Dahl_2009,Park_2010,Mueller_2010,Mueller_2011}.

\subsection{Basic properties}
\label{section:basic}

We list here some basic properties of the MFM model. See Appendix \ref{section:proofs} for proofs.
Denoting $x_c =(x_j: j\in c)$ and
$m(x_c) =\int_\Theta \big[\prod_{j\in c} f_\theta(x_j)\big]\,H(d\theta)$ (with the convention that $m(x_\emptyset) = 1$), we have
\begin{align}\label{equation:marginal-product}
p(x_{1:n}|\C) = \prod_{c\in\C} m(x_c).
\end{align}
The number of components $K$ and the number of clusters $T =|\C|$ are related by
\begin{align}
p(t|k) &= \frac{k_{(t)}}{(\gamma k)^{(n)}}\sum_{\C:|\C|= t}\prod_{c\in\C}\gamma^{(|c|)},\label{equation:ptk}\\
p(k|t) &= \frac{1}{V_n(t)}\frac{k_{(t)}}{(\gamma k)^{(n)}}\,p_K(k),\label{equation:pkt}
\end{align}
where in Equation \ref{equation:ptk}, the sum is over partitions $\C$ of $[n]$ such that $|\C|= t$.
The formula for $p(k|t)$ is required for doing inference about the number of components $K$ based on posterior samples of $\C$; fortunately, it is easy to compute.
We have the conditional independence relations
\begin{align}
\C & \perp K\mid T,\label{equation:CKT}\\
X_{1:n} & \perp K\mid T.\label{equation:XKT}
\end{align}

\subsection{The coefficients $V_n(t)$}
\label{section:v}

The following recursion is a special case of a more general result for Gibbs partitions \citep{Gnedin_2006}.

\begin{proposition}
The numbers $V_n(t)$ (Equation \ref{equation:v}) satisfy the recursion
\begin{align}\label{equation:recursion}
V_{n +1}(t+1) = V_n(t)/\gamma-(n/\gamma + t) V_{n +1}(t)
\end{align}
for any $0\leq t\leq n$ and $\gamma>0$.
\end{proposition}

This is easily seen by plugging the identity
$$k_{(t+1)} = (\gamma k + n)k_{(t)}/\gamma-(n/\gamma + t) k_{(t)} $$
into the expression for $V_{n+1}(t+1)$.
In the case of $\gamma = 1$, \cite{Gnedin_2010} has discovered a beautiful example of a distribution on $K$ for which both $p_K(k)$ and $V_n(t)$ have closed-form expressions.

In previous work on the MFM model, it has been common for $p_K$ to be chosen to be proportional to a Poisson distribution restricted to the positive integers or a subset thereof \citep{Phillips_1996,Stephens_2000,Nobile_2007}, and \cite{Nobile_2005} has proposed a theoretical justification for this choice.
Interestingly, the model has some nice mathematical properties if one instead chooses $K-1$ to be given a Poisson distribution, that is, $p_K(k) =\mathrm{Poisson}(k-1|\lambda)$ for some $\lambda>0$. One example of this arises here (for another example, see Section \ref{section:stick}): it turns out that if $p_K(k) =\mathrm{Poisson}(k-1|\lambda)$ and $\gamma= 1$ then
\begin{align}\label{equation:Poisson-V_n}
V_n(0) = \frac{1}{\lambda^n}\Big(1-\sum_{k=1}^n p_K(k)\Big).
\end{align}

However, to do inference, it is not necessary to choose $p_K$ to have any particular form. To do inference, we just need to be able to compute $p(\C)$, and in turn, we need to be able to compute $V_n(t)$. To this end, note that $k_{(t)}/(\gamma k)^{(n)}\leq k^t/(\gamma k)^n$, and thus the infinite series for $V_n(t)$ converges rapidly when $t\ll n$. It always converges to a finite value when $1\leq t\leq n$; this is clear from the fact that $p(\C)\in[0,1]$. This finiteness can also be seen directly from the series since $k^t/(\gamma k)^n \leq 1/\gamma^n$, and in fact, this shows that the series for $V_n(t)$ converges at least as rapidly (up to a constant) as the series $\sum_{k = 1}^\infty p_K(k)$ converges to $1$. 
Hence, for any reasonable choice of $p_K$ (i.e., not having an extraordinarily heavy tail), $V_n(t)$ can easily be numerically approximated to a high level of precision. In practice, computing the required values of $V_n(t)$ takes a negligible amount of time.



\subsection{Self-consistent marginals}
\label{section:self-consistent}

For each $n = 1,2,\dotsc$, let $q_n(\C)$ denote the distribution on partitions of $[n]$ as defined above (Equation \ref{equation:EPPF}).  This family of partition distributions is preserved under marginalization, in the following sense.

\begin{proposition}
\label{proposition:self-consistent}
If $m<n$ then $q_m$ coincides with the marginal distribution on partitions of $[m]$ induced by $q_n$.
\end{proposition}

In other words, drawing a sample from $q_n$ and removing elements $m+1,\dotsc,n$ from it yields a sample from $q_m$. This can be seen directly from the model definition (Equation \ref{equation:model}), since $\C$ is the partition induced by the $Z$'s, and the distribution of $Z_{1:m}$ is the same when the model is defined with any $n\geq m$. This property is sometimes referred to as \emph{consistency in distribution} \citep{Pitman_2006}.

By Kolmogorov's extension theorem (e.g., \citealp{Durrett_1996}), it is well-known that this implies the existence of a unique probability distribution on partitions of the positive integers $\Z_{>0}=\{1,2,\dotsc\}$ such that the marginal distribution on partitions of $[n]$ is $q_n$ for all $n\in\{1,2,\dotsc\}$.
A random partition of $\Z_{>0}$ from such a distribution is a \emph{combinatorial stochastic process}; for background, see \cite{Pitman_2006}.

\section{Restaurant process, stick-breaking, and random measure representations}
\label{section:representations}

\subsection{P\'{o}lya urn scheme / Restaurant process}
\label{section:restaurant}

\cite{Pitman_1996} considered a general class of urn schemes, or restaurant processes, corresponding to exchangeable partition probability functions (EPPFs). The following scheme for the MFM falls into this general class.

\begin{theorem}
\label{theorem:restaurant}
The following process generates partitions $\C_1,\C_2,\dotsc$ such that for any $n\in\{1,2,\dotsc\}$, the probability mass function of $\C_n$ is given by Equation \ref{equation:EPPF}.
\begin{itemize}
\item Initialize with a single cluster consisting of element $1$ alone: $\C_1 =\{\{1\}\}$.
\item For $n = 2,3,\dotsc$, element $n$ is placed in \dots
\begin{itemize}
\item[] an existing cluster $c\in\C_{n-1}$ with probability $\propto |c|+\gamma $
\item[] a new cluster with probability $\displaystyle\propto \frac{V_n(t+1)}{V_n(t)}\gamma$
\end{itemize}
where $t =|\C_{n-1}|$.
\end{itemize}
\end{theorem}

Clearly, this bears a close resemblance to the Chinese restaurant process (i.e., the Blackwell--MacQueen urn process), in which the $n$th element is placed in an existing cluster $c$ with probability $\propto |c|$ or a new cluster with probability $\propto \alpha$ (the concentration parameter) \citep{Blackwell_1973,aldous1985exchangeability}.


\subsection{Random discrete measures}
\label{section:discrete-measure}

The MFM can also be formulated starting from a distribution on discrete measures that is analogous to the Dirichlet process. With $K$, $\pi$, and $\theta_{1:K}$ as in Equation \ref{equation:model}, let
$$ G =\sum_{i = 1}^K \pi_i \delta_{\theta_i} $$
where $\delta_\theta$ is the unit point mass at $\theta$. Let us denote the distribution of $G$ by $\M(p_K,\gamma,H)$. Note that $G$ is a random discrete measure over $\Theta$. If $H$ is continuous (i.e., $H(\{\theta\}) = 0$ for all $\theta\in\Theta$), then with probability $1$, the number of atoms is $K$; otherwise, there may be fewer than $K$ atoms.
If we take $X_1,\dotsc,X_n|G$ i.i.d.\ from the resulting mixture, namely, 
$$f_G(x):=\int f_\theta(x)G(d\theta) = \sum_{i = 1}^K \pi_i f_{\theta_i}(x), $$
then the distribution of $X_{1:n}$ is the same as before.
So, in this notation, the MFM model is:
\begin{align*}
& G\sim\M(p_K,\gamma,H)\\
& X_1,\dotsc,X_n\iid f_G, \mbox{ given $G$.}
\end{align*}

This random discrete measure perspective is connected to work on \emph{species sampling models} \citep{Pitman_1996} in the following way. When $H$ is continuous, we can construct a species sampling model by letting $G\sim\M(p_K,\gamma,H)$ and modeling the observed data as $\beta_1,\dotsc,\beta_n\sim G$. We refer to \cite{Pitman_1996}, \cite{Hansen_2000}, \cite{Ishwaran_2003}, \cite{Lijoi_2005b,Lijoi_2007}, and \cite{Lijoi_2008} for more background on species sampling models and further examples.
The following posterior predictive rule for this particular model follows the form of Pitman's general rule; note the close relationship to the restaurant process (Theorem \ref{theorem:restaurant} above).

\begin{theorem}
\label{theorem:species-predictive}
If $H$ is continuous, then $\beta_1\sim H$ and the distribution of $\beta_n$ given $\beta_1,\dotsc,\beta_{n-1}$ is proportional to
\begin{align}\label{equation:species-predictive}
\frac{V_n(t+1)}{V_n(t)}\gamma H + \sum_{i = 1}^t (n_i+\gamma)\delta_{\beta_i^*},
\end{align}
where $\beta^*_1,\dotsc,\beta^*_t$ are the distinct values taken by $\beta_1,\dotsc,\beta_{n-1}$, and $n_i =\#\big\{j\in[n-1]:\beta_j =\beta_i^*\big\}$.
\end{theorem}

For comparison, when $G\sim\mathrm{DP}(\alpha H)$ instead, the distribution of $\beta_n$ given $\beta_1,\dotsc,\beta_{n-1}$ is proportional to $\alpha H + \sum_{j = 1}^{n-1} \delta_{\beta_j} = \alpha H +\sum_{i = 1}^t n_i\beta_i^*$, since $G|\beta_1,\dotsc,\beta_{n-1}\sim\mathrm{DP}(\alpha H +\sum_{j = 1}^{n-1}\delta_{\beta_j})$ \citep{Ferguson_1973,Blackwell_1973}.


\subsection{Stick-breaking representation}
\label{section:stick}

The Dirichlet process has an elegant stick-breaking representation for the mixture weights $\pi_1,\pi_2,\dotsc$ \citep{Sethuraman_1994,Sethuraman_1981}. This extraordinarily clarifying perspective has inspired a number of  nonparametric models \citep{MacEachern_1999,MacEachern_2000,Hjort_2000,Ishwaran_2000,Ishwaran_2001b,Griffin_2006,Dunson_2008,Chung_2009,Rodriguez_2011,Broderick_2012}, has provided insight into the properties of related models \citep{Favaro_2012,Teh_2007,Thibaux_2007,Paisley_2010}, and has been used to develop efficient inference algorithms \citep{Ishwaran_2001b,Blei_2006,Papaspiliopoulos_2008,Walker_2007b,Kalli_2011}.

In a certain special case---namely, when $p_K(k) =\mathrm{Poisson}(k-1|\lambda)$ and $\gamma = 1$---we have noticed that the MFM also has an interesting representation that we describe using the stick-breaking analogy, although it is somewhat different in nature. This is another example of the nice mathematical properties resulting from this choice of $p_K$ and $\gamma$.
Consider the following procedure: 
\begin{quote}
\emph{Take a unit-length stick, and break off i.i.d.\ $\mathrm{Exponential}(\lambda)$ pieces until you run out of stick.}
\end{quote}
In other words, let $\epsilon_1,\epsilon_2,\dotsc\iid\mathrm{Exponential}(\lambda)$, define $\tilde K =\min\{j:\sum_{i = 1}^j\epsilon_i \geq 1\}$, and set $\tilde\pi_i =\epsilon_i$ for $i = 1,\dotsc,\tilde K - 1$ and $\tilde\pi_{\tilde K} = 1-\sum_{i = 1}^{\tilde K-1}\tilde\pi_i$.

\begin{proposition}
\label{proposition:stick}
The stick lengths $\tilde\pi$ have the same distribution as the mixture weights $\pi$ in the MFM model when $p_K(k) =\mathrm{Poisson}(k-1|\lambda)$ and $\gamma = 1$.
\end{proposition}

This is a consequence of a standard construction for Poisson processes. This suggests a way of generalizing the MFM model: take any sequence of nonnegative random variables $(\epsilon_1,\epsilon_2,\dotsc)$ (not necessarily independent or identically distributed) such that $\sum_{i = 1}^\infty\epsilon_i>1$ with probability $1$, and define $\tilde K$ and $\tilde\pi$ as above. Although the distribution of $\tilde K$ and $\tilde\pi$ may be complicated, in some cases it might still be possible to do inference based on the stick-breaking representation. This might be an interesting way to introduce different kinds of prior information on the mixture weights, however, we have not explored this possibility.



\section{Asymptotics}
\label{section:asymptotics}

In this section, we consider the asymptotics of $V_n(t)$, the asymptotic relationship between the number of components and the number of clusters, and the approximate form of the conditional distribution on cluster sizes given the number of clusters.

\subsection{Asymptotics of $V_n(t)$}
\label{section:v-asymptotics}

Recall that $V_n(t) =\sum_{k = 1}^\infty\frac{k_{(t)}}{(\gamma k)^{(n)}}\,p_K(k)$ (Equation \ref{equation:v})
for $1\leq t\leq n$, with $\gamma>0$ and $p_K$ a p.m.f.\ on $\{1,2,\dotsc\}$.

\begin{theorem}
\label{theorem:v-asymptotics}
For any $t\in\{1,2,\dotsc\}$, if $p_K(t)>0$ then
\begin{align}
V_n(t)\sim\frac{t_{(t)}}{(\gamma t)^{(n)}}\,p_K(t) \sim \frac{t!}{n!}\,\frac{\Gamma(\gamma t)}{n^{\gamma t - 1}}\,p_K(t)
\end{align}
as $n\to\infty$.
\end{theorem}

In particular, $V_n(t)$ has a simple interpretation, asymptotically---it behaves like the $k = t$ term in the series.

\subsection{Relationship between the number of clusters and number of components}
\label{section:clusters-components-relationship}

In the MFM, it is perhaps intuitively clear that, under the prior at least, the number of clusters $T =|\C|$ behaves very similarly to the number of components $K$ when $n$ is large. It turns out that under the posterior they also behave very similarly for large $n$. 

\begin{theorem}
\label{theorem:clusters-components-relationship}
Let $x_1,x_2,\dotsc\in\X$ and $k\in\{1,2,\dotsc\}$. If $p_K(1),\dotsc,p_K(k)>0$ then
$$ \big|p(T = k\mid x_{1:n}) - p(K = k\mid x_{1:n})\big| \longrightarrow 0 $$
as $n\to\infty$.
\end{theorem}

\subsection{Distribution of the cluster sizes under the prior}
\label{section:part-sizes}

Here, we examine one of the major differences between the MFM and DPM priors. Roughly speaking, under the prior, the MFM prefers all clusters to be the same order of magnitude, while the DPM prefers having a few large clusters and many very small clusters.  In the following calculations, we quantify the preceding statement more precisely. (See \cite{Green_2001} for informal observations along these lines.) Interestingly, these prior influences remain visible in certain aspects of the posterior, even in the limit as $n$ goes to infinity, as shown by the inconsistency of DPMs for the number of components in a finite mixture \citep{Miller_2014}.

Let $\C$ be the partition of $[n]$ in the MFM model (Equation \ref{equation:EPPF}), and let $A =(A_1,\dotsc,A_T)$ be the ordered partition of $[n]$ obtained by randomly ordering the parts of $\C$, uniformly among the $T!$ possible choices, where $T =|\C|$. Then
$$ p(A) = \frac{p(\C)}{|\C|!} =\frac{1}{t!} V_n(t)\prod_{i = 1}^t\gamma^{(|A_i|)},$$
where $t =|\C|$. 
Now, let $S =(S_1,\dotsc,S_T)$ be the vector of part sizes of $A$, that is, $S_i =|A_i|$. Then
$$ p(S = s) =\sum_{A: S(A) = s} p(A) = V_n(t)\,\frac{n!}{t!}\prod_{i = 1}^t\frac{\gamma^{(s_i)}}{s_i!} $$
for $s\in\Delta_t$, $t\in\{1,\ldots,n\}$, where $\Delta_t =\{s\in\Z^t:\sum_i s_i = n,\,s_i\geq 1\,\forall i\}$ (i.e., the $t$-part compositions of $n$).
For any $x>0$, writing $x^{(m)}/m! =\Gamma(x + m)/(m!\,\Gamma(x))$ and using Stirling's approximation, we have
$x^{(m)} / m! \sim m^{x - 1} / \Gamma(x)$ as $m\to\infty$. This yields the approximations
$$ p(S = s) \approx \frac{V_n(t)}{\Gamma(\gamma)^t}\,\frac{n!}{t!}\,\prod_{i = 1}^t s_i^{\gamma - 1} 
\approx \frac{p_K(t)}{n^{\gamma t-1}}\,\frac{\Gamma(\gamma t)}{\Gamma(\gamma)^t}\,\prod_{i = 1}^t s_i^{\gamma - 1}$$
(using Theorem \ref{theorem:v-asymptotics} in the second step), and
$$ p(S = s\mid T = t)\approx \kappa\prod_{i = 1}^t s_i^{\gamma - 1} $$
for $s\in\Delta_t$, where $\kappa$ is a normalization constant. Thus $p(s|t)$, although a discrete distribution, has approximately the same shape as a symmetric $t$-dimensional Dirichlet distribution. This would be obvious if we were conditioning on the number of components $K$, and it makes intuitive sense when conditioning on $T$, since $K$ and $T$ are essentially the same for large $n$.

It is very interesting to compare this to the corresponding distributions for Dirichlet process mixtures. In the  DPM, we have $p_\DPM(\C) =\frac{\alpha^t}{\alpha^{(n)}}\prod_{c\in\C}(|c|- 1)!$, and $p_\DPM(A) = p_\DPM(\C)/|\C|!$ as before, so for $s\in\Delta_t$, $t\in\{1,\ldots,n\}$,
$$ p_\DPM(S =s) =  \frac{n!}{\alpha^{(n)}}\,\frac{\alpha^t}{t!}\,s_1^{-1}\cdots s_t^{-1}$$
and
$$ p_\DPM(S = s\mid T = t)\propto s_1^{-1}\cdots s_t^{-1}, $$
which has the same shape as a $t$-dimensional Dirichlet distribution with all the parameters taken to $0$ (noting that this is normalizable since $\Delta_t$ is finite). Asymptotically in $n$, $p_\DPM(s|t)$ puts all of its mass in the ``corners'' of the discrete simplex $\Delta_t$, while under the MFM, $p(s|t)$ remains more evenly dispersed.

\section{Inference algorithms}
\label{section:inference}

As shown by the results of Sections \ref{section:partitions} and \ref{section:representations}, MFMs have many of the same properties as DPMs. As a result, much of the extensive body of work on MCMC samplers for DPMs can be directly applied to MFMs, including samplers for conjugate and non-conjugate cases, as well as split-merge samplers.

When $H$ is a conjugate prior for $\{f_\theta\}$, such that the marginal likelihood $m(x_c) =\int_\Theta \big[\prod_{j\in c} f_\theta(x_j)\big]\,H(d\theta)$ can be easily computed, the following Gibbs sampling algorithm can be used to sample from the posterior on partitions, $p(\C|x_{1:n})$. Given a partition $\C$, let $\C\setminus j$ denote the partition obtained by removing element $j$ from $\C$.
\begin{enumerate}
\item Initialize $\C =\{[n]\}$ (i.e., one cluster). 
\item Repeat the following $N$ times, to obtain $N$ samples.
\begin{itemize}
\item[] For $j = 1,\dotsc,n$:  Remove element $j$ from $\C$, and place it \dots
\begin{itemize}
\item[] in $c\in\C\setminus j$ with probability $\displaystyle\propto (|c|+\gamma)\frac{m(x_{c\,\cup j})}{m(x_c)}$
\item[] in a new cluster with probability $\displaystyle\propto \gamma\, \frac{V_n(t+1)}{V_n(t)}\,m(x_j)$
\end{itemize}
where $t=|\C\setminus j|$.
\end{itemize}
\end{enumerate}
This is a direct adaptation of ``Algorithm 3'' for DPMs \citep{MacEachern_1994,Neal_1992,Neal_2000}. The only differences are that in Algorithm 3, $|c|+\gamma$ is replaced by $|c|$, and $\gamma V_n(t+1)/V_n(t)$ is replaced by $\alpha$ (the concentration parameter). Thus, the differences between the MFM and DPM versions of the algorithm are precisely the same as the differences between their respective restaurant processes.
Computing the required values of $V_n(t)$ takes a negligible amount of time compared to running the sampler.
In order for the algorithm to be valid, the Markov chain needs to be irreducible, and to achieve this it is necessary to have $\big\{t\in\{1,2,\dotsc\}:V_n(t)>0\big\}$ be a block of consecutive integers. In fact, it turns out that this is always the case (and this block includes $t = 1$), since for any $k$ such that $p_K(k)>0$, we have $V_n(t)>0$ for all $t = 1,\dotsc,k$.

When $H$ is a non-conjugate prior (and $m(x_c)$ cannot be easily computed), a clever auxiliary variable technique referred to as ``Algorithm 8'' can be used for inference in the DPM \citep{Neal_2000,MacEachern_1998}. Making the same substitutions as above, we can apply Algorithm 8 to perform inference in the MFM as well; see \citet{Miller_thesis} for details.

A well-known issue with incremental Gibbs samplers such as these, however, when applied to DPMs, is that the mixing can be somewhat slow, since it may take a long time to create or destroy substantial clusters by moving one element at a time. With MFMs, this issue seems to be exacerbated, since MFMs tend to put small probability (compared with DPMs) on partitions with tiny clusters (see Section \ref{section:part-sizes}), making it difficult for the sampler to move through these regions of the space. 

To deal with this issue, split-merge samplers for DPMs have been developed, in which a large number of elements can be reassigned in a single move \citep{Dahl_2003,Dahl_2005,Jain_2004,Jain_2007}. In the same way as the incremental samplers, one can directly apply these split-merge samplers (both conjugate and non-conjugate) to MFMs, using the properties described in Sections \ref{section:partitions} and \ref{section:representations}. More generally, it seems likely that any partition-based MCMC sampler for DPMs could be applied to MFMs as well.

In Section \ref{section:empirical}, we apply the Jain--Neal split-merge samplers coupled with incremental Gibbs samplers, in both conjugate and non-conjugate settings.


\section{Empirical demonstrations}
\label{section:empirical}


In this section, we demonstrate the MFM on simulated and real datasets. All of the examples below involve Gaussian component densities, but of course our approach is not limited to mixtures of Gaussians.

\subsection{Simulation example}
\label{section:simulation}

In the introduction, we presented several figures comparing the MFM and DPM on data from a three-component bivariate Gaussian mixture, illustrating the behavior of the MFM with respect to density estimation, clustering, and inference for the number of components. Here, we provide the details of the data, model, and method of inference for this simulation example.

\subsubsection*{Data}

The data distribution is 
$\sum_{i = 1}^3 w_i \N(\mu_i,C_i)$ where $w =(0.45,0.3,0.25)$, $\mu_1 =\matrixsmall{4\\4}$, $\mu_2 =\matrixsmall{7\\4}$, $\mu_3 =\matrixsmall{6\\2}$, 
$C_1 = \matrixsmall{1 & 0\\ 0 & 1}$, 
$C_2 = R\matrixsmall{2.5 & 0\\ 0 & 0.2}R^\T$ where 
$R =\matrixsmall{\cos\rho & \, -\sin\rho \\ \sin\rho & \cos\rho}$ with $\rho=\pi/4$, and
$C_3 = \matrixsmall{3 & 0\\ 0 & 0.1}$.

\subsubsection*{Model}

The component densities are multivariate normal,
$f_\theta(x) = f_{\mu,\Lambda}(x) =\N(x|\mu,\Lambda^{-1})$
and the base measure (prior) $H$ on $\theta =(\mu,\Lambda)$ is
$\mu\sim\N(\widehat\mu,\widehat C)$, $\Lambda\sim\Wishart_d(V,\nu)$ 
independently, where $\widehat\mu$ is the sample mean, $\widehat C$ is the sample covariance, $\nu=d=2$, and $V=\widehat C^{-1}/\nu$.  Here, $\Wishart_d(\Lambda|V,\nu)\propto |\det \Lambda|^{(\nu-d-1)/2} \exp\big(-\tfrac{1}{2} \mathrm{tr}(V^{-1}\Lambda)\big)$. Note that this is a data-dependent prior.

For the MFM, we take $K\sim\mathrm{Geometric}(r)$ ($p_K(k) = (1-r)^{k-1} r$ for $k = 1,2,\dotsc$) with $r = 0.1$, and we choose $\gamma = 1$ for the finite-dimensional Dirichlet parameters.
For the DPM, we put an $\mathrm{Exponential}(1)$ prior on the concentration parameter, $\alpha$.

Note that taking $\mu$ and $\Lambda$ to be independent results in a non-conjugate prior. This prior is appropriate when the location of the data is not informative about the covariance (and vice versa).

\subsubsection*{Inference}

For both the MFM and DPM, we use the non-conjugate split-merge sampler of \citet{Jain_2007}, coupled with Algorithm 8 of \citet{Neal_2000} (using a single auxiliary variable) for incremental Gibbs updates to the partition. Specifically, following \citet{Jain_2007}, we use the (5,1,1,5) scheme: 5 intermediate scans to reach the split launch state, 1 split-merge move per iteration, 1 incremental Gibbs scan per iteration, and 5 intermediate moves to reach the merge launch state. Gibbs updates to the DPM concentration parameter $\alpha$ are made using Metropolis--Hastings moves.

Five independent datasets were used for each $n \in\{50,100,250,1000\}$, and for each model (MFM and DPM), the sampler was run for 5,000 burn-in iterations and 95,000 sample iterations (for a total of 100,000). Judging by traceplots and running averages of various statistics, this appeared to be sufficient for mixing. The cluster sizes were recorded after each iteration, and to reduce memory storage requirements, the full state of the chain was recorded only once every 100 iterations. For each run, the seed of the random number generator was initialized to the same value for both the MFM and DPM.

For a dataset of size $n$, the sampler used for these experiments took approximately $8\times 10^{-6}\,n$ seconds per iteration, using a 2.80 GHz processor with 6 GB of RAM.

\subsubsection*{Results}

As described in the introduction, the results of this simulation empirically indicate that on data from a finite mixture, MFMs and DPMs are consistent for the density (Figures \ref{figure:density-estimates} and \ref{figure:Hellinger}), DPM clusterings tend to have small extra clusters while MFM clusterings do not (Figure \ref{figure:clustering}), and MFMs are consistent for the number of components while DPMs are not (Figure \ref{figure:tk-posteriors}).  This is what we expect from theory (although to be precise, the inconsistency result of \citet{Miller_2014} only applies to the case of fixed concentration parameter $\alpha$).
These results are not too surprising, since when the data distribution is a finite mixture from the assumed family, the MFM is correctly specified, while the DPM is not.  On data from an infinite mixture, one would expect the DPM to have certain advantages.
See Appendix \ref{section:formulas} for formulas for computing the posterior on $k$ and the density estimates.

\subsection{Galaxy dataset}
\label{section:galaxy}

The galaxy dataset \citep{Roeder_1990} is a standard benchmark for mixture models, consisting of measurements of the velocities of $82$ galaxies in the Corona Borealis region; see Figure \ref{figure:galaxy}.
The purpose of this example is to demonstrate agreement between our method and published results using reversible jump MCMC with the same model, and also to show that using hyperpriors presents no difficulties.

\subsubsection*{Model}

To enable comparison, we use exactly the same model as \citet{Richardson_1997}.
The component densities are univariate normal,
$f_\theta(x) = f_{\mu,\lambda}(x) =\N(x|\mu,\lambda^{-1})$,
and the base measure $H$ on $\theta =(\mu,\lambda)$ is
$\mu\sim\N(\mu_0,\sigma_0^2)$, $\lambda\sim\Ga(a,b)$
independently (where $\Ga(\lambda|a,b)\propto \lambda^{a-1} e^{-b \lambda}$). Further, a hyperprior is placed on $b$, by taking $b\sim\Ga(a_0,b_0)$. The remaining parameters are set to
$\mu_0 =(\max\{x_i\}+\min\{x_i\})/2$, $\sigma_0 =\max\{x_i\}-\min\{x_i\}$, $a = 2$, $a_0 = 0.2$, and $b_0 = 10/\sigma_0^2$. Note that the parameters $\mu_0$, $\sigma_0$, and $b_0$ are functions of the observed data $x_1,\dotsc,x_n$.
See \citet{Richardson_1997} for the rationale behind these parameter choices. (Note: This choice of $\sigma_0$ may be a bit too large, affecting the posteriors on the number of clusters and components, however, we stick with it to enable comparisons to \citet{Richardson_1997}.)
For the MFM, following \citet{Richardson_1997}, we take $K\sim\mathrm{Uniform}\{1,\allowbreak\dotsc,30\}$ and $\gamma = 1$. For the DPM, we take $\alpha\sim\mathrm{Exponential}(1)$.

\subsubsection*{Inference}

As before, we use the non-conjugate split-merge sampler of \citet{Jain_2007} coupled with Algorithm 8 of \citet{Neal_2000}, and Gibbs updates to the DPM concentration parameter $\alpha$ are made using Metropolis--Hastings.
We use Gibbs sampling to handle the hyperprior on $b$ (i.e., append $b$ to the state of the Markov chain, run the sampler given $b$ as usual, and periodically sample $b$ given everything else). More general hyperprior structures can be handled similarly. In all other respects, the same inference algorithm as in Section \ref{section:simulation} was used.

We do not restrict the parameter space in any way (e.g., forcing the component means to be ordered to obtain identifiability, as was done by \citet{Richardson_1997}). All of the quantities we consider are invariant to the labeling of the clusters. See \citet{Jasra_2005} for discussion on this point.

The sampler was run for 5,000 burn-in iterations, and 45,000 sample iterations.  This appeared to be more than sufficient for mixing. Cluster sizes were recorded after each iteration, and the full state of the chain was recorded every 50 iterations.
Each iteration took approximately $8\times 10^{-6}\,n$ seconds, with $n = 82$.

\subsubsection*{Results}

Figure \ref{figure:galaxy} shows the estimated densities and the posteriors on the number of clusters and components.  Comparing this with Figure 2(c) of \citet{Richardson_1997}, we see that our MFM density estimate is visually indistinguishable from theirs (as it should be, since we are using the same model with the same parameters).

\begin{figure}
\centering
\includegraphics[trim=0.8cm 0.5cm 1cm 0.2cm, clip=true, width=0.49\textwidth]{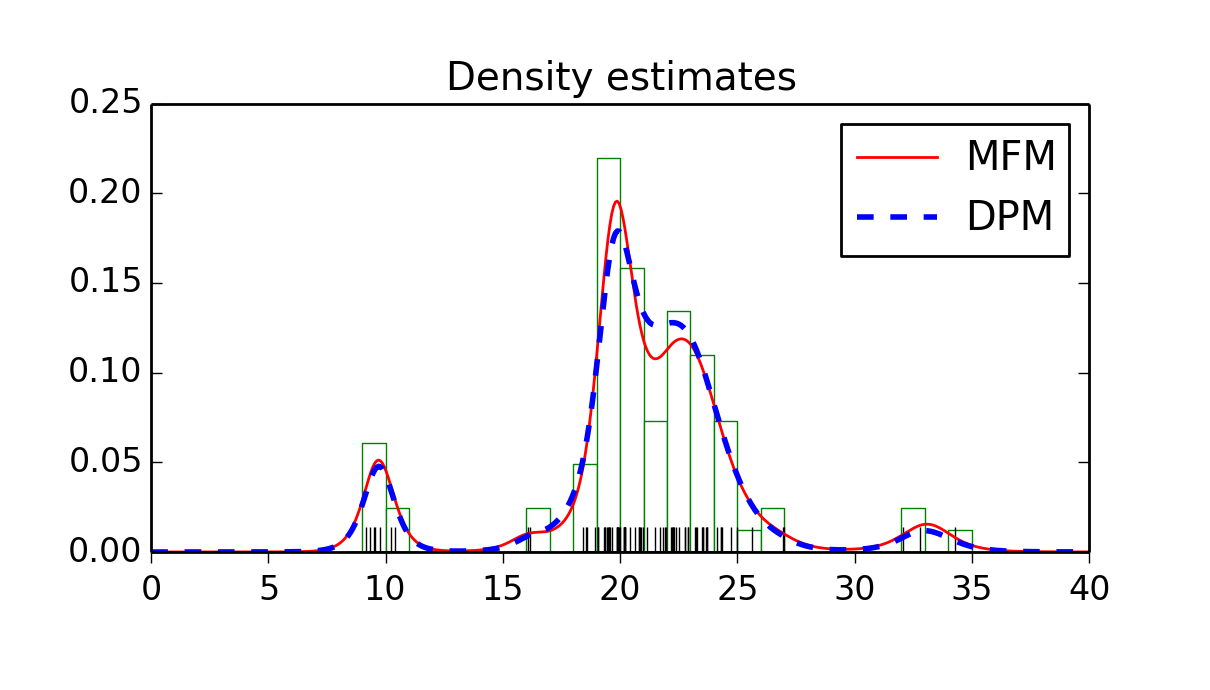}
\includegraphics[trim=0.8cm 0.5cm 1cm 0.2cm, clip=true, width=0.49\textwidth]{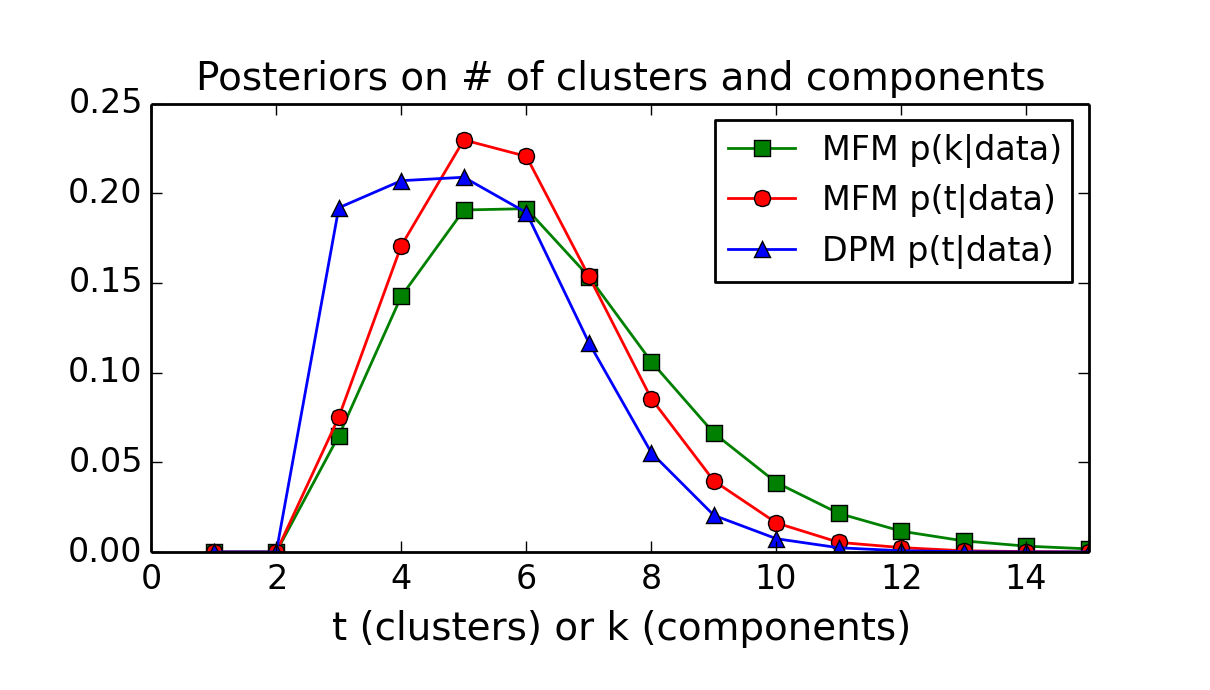}
\caption{Results on the galaxy dataset. Left: Histogram of the data (green bars), rug plot of the data (black ticks), and estimated densities using the MFM (red solid line) and DPM (blue dashed line). Right: MFM and DPM posteriors on the number of clusters ($t$), along with the MFM posterior on the number of components ($k$).}
\label{figure:galaxy}
\end{figure}

Table \ref{table:galaxy} compares our estimate of the MFM posterior on the number of components $k$ with the results of \citet{Richardson_1997}. Again, the results are very close, as expected.

\begin{table}[!ht]
\centering
\caption{Estimate of the MFM posterior on $k$ for the galaxy dataset.}%
\label{table:galaxy}%
\begin{tabular}{|c|c|c|c|c|c|c|c|}%
\hline
$k$ & 1 & 2 & 3 & 4 & 5 & 6 & 7 \\
\hline
Here & 0.000 & 0.000 & 0.065 & 0.143 & 0.191 & 0.191 & 0.153 \\
R\&G & 0.000 & 0.000 & 0.061 & 0.128 & 0.182 & 0.199 & 0.160 \\
\hline
\end{tabular}\\[1em]
\begin{tabular}{|c|c|c|c|c|c|c|c|}%
\hline
8 & 9 & 10 & 11 & 12 & 13 & 14 & 15 \\
\hline
0.106 & 0.066 & 0.039 & 0.021 & 0.012 & 0.006 & 0.003 & 0.002 \\
0.109 & 0.071 & 0.040 & 0.023 & 0.013 & 0.006 & 0.003 & 0.002 \\
\hline
\end{tabular}
\end{table}

\subsection{Discriminating cancer types using gene expression data}
\label{section:gene-expression}

In cancer research, gene expression profiling---that is, measuring the degree to which each gene is expressed by a given tissue sample under given conditions---enables the identification of distinct subtypes of cancer, leading to greater understanding of the mechanisms underlying cancers as well as potentially providing patient-specific diagnostic tools.  In gene expression datasets, there are typically a small number of very high-dimensional data points, each consisting of the gene expression levels in a given tissue sample under given conditions.

One approach to analyzing gene expression data is to use Gaussian mixture models to identify clusters which may represent distinct cancer subtypes \citep{Yeung_2001,mclachlan2002mixture,Medvedovic_2002,Medvedovic_2004,deSouto_2008,Rasmussen_2009,McNicholas_2010}. In fact, in a comparative study of seven clustering methods on 35 cancer gene expression datasets with known ground truth, \citet{deSouto_2008} found that finite mixtures of Gaussians provided the best results---when the number of components $k$ was set to the true value.  However, in practice, choosing an appropriate value of $k$ can be difficult.  Using the methods developed in this paper, the MFM provides a principled approach to inferring the clusters even when $k$ is unknown, as well as doing inference for $k$, provided that the components are well-modeled by Gaussians. (However, see Section \ref{section:discussion} for some potential pitfalls.)

The purpose of this example is to demonstrate that our approach can work well even in very high-dimensional settings, and may provide a useful tool for this application. It should be emphasized that we are not cancer scientists, so the results reported here should not be interpreted as scientifically relevant, but simply as a proof-of-concept.

\subsubsection*{Data}

We apply the MFM to gene expression data collected by \citet{Armstrong_2001} in a study of leukemia subtypes. \citet{Armstrong_2001} measured gene expression levels in samples from 72 patients who were known to have one of two leukemia types, acute lymphoblastic leukemia (ALL) or acute myelogenous leukemia (AML), and they found that a previously undistinguished subtype of ALL, which they termed mixed-lineage leukemia (MLL), could be distinguished from conventional ALL and AML based on the gene expression profiles.

We use the preprocessed data provided by \citet{deSouto_2008}, which they filtered to include only genes with expression levels differing by at least 3-fold in at least 30 samples, relative to their mean expression level across all samples.  The resulting dataset consists of 72 samples and 1081 genes per sample, i.e., $n = 72$ and $d = 1081$. Following standard practice, we take the base-2 logarithm of the data before analysis, and normalize each dimension to have zero mean and unit variance.

\subsubsection*{Model}

For simplicity, we use multivariate Gaussian component densities with diagonal covariance matrices, i.e., the dimensions are independent univariate Gaussians, and we place independent conjugate priors on each dimension.  Thus, for each component, for $i = 1,\dotsc,d$, dimension $i$ is $\N(\mu_i,\lambda_i^{-1})$,
with $\lambda_i\sim\Ga(a,b)$ and $\mu_i|\lambda_i\sim\N(0,(c\lambda_i)^{-1})$.
We choose $a=1$, $b=1$, and $c=1$. (Recall that the data is zero mean, unit variance in each dimension.)
For the MFM, $K\sim\mathrm{Geometric}(0.1)$ and $\gamma = 1$, and for the DPM, $\alpha\sim\mathrm{Exponential}(1)$.
These are all simply default settings and have not been tailored to the problem; a careful scientific investigation would involve thorough prior elicitation, sensitivity analysis, and model checking.

\subsubsection*{Inference}

Given the partition $\C$ of the data into clusters, the parameters can be integrated out analytically since the prior is conjugate. Thus, for both the MFM and DPM, we use the split-merge sampler of \citet{Jain_2004} for conjugate priors, coupled with Algorithm 3 of \citet{Neal_2000}. Following \citet{Jain_2004}, we use the (5,1,1) scheme: 5 intermediate scans to reach the split launch state, 1 split-merge move per iteration, and 1 incremental Gibbs scan per iteration.

Due to the high-dimensionality of the parameters, this has far better mixing time than sampling the parameters, as is done in reversible jump MCMC.

The sampler was run for 1,000 burn-in iterations, and 19,000 sample iterations.  This appears to be many more iterations than required for burn-in and mixing in this particular example---in fact, only 5 to 10 iterations are required to separate the clusters, and the results are indistinguishable when using only 10 burn-in and 190 sample iterations.
The full state of the chain was recorded every 20 iterations.
Each iteration took approximately $1.3\times 10^{-3}\,n$ seconds, with $n = 72$.

\subsubsection*{Results}

The posteriors on the number of clusters $t$ are concentrated at $3$ (see Figure \ref{figure:gene}), in agreement with the division into ALL, MLL, and AML determined by \citet{Armstrong_2001}. The MFM posterior on $k$ is shifted slightly to the right because there are a small number of observations; this accounts for uncertainty regarding the possibility of additional components that were not observed in the given data.

\begin{figure}
\centering
\includegraphics[trim=0.8cm 0.5cm 1cm 0.2cm, clip=true, width=0.49\textwidth]{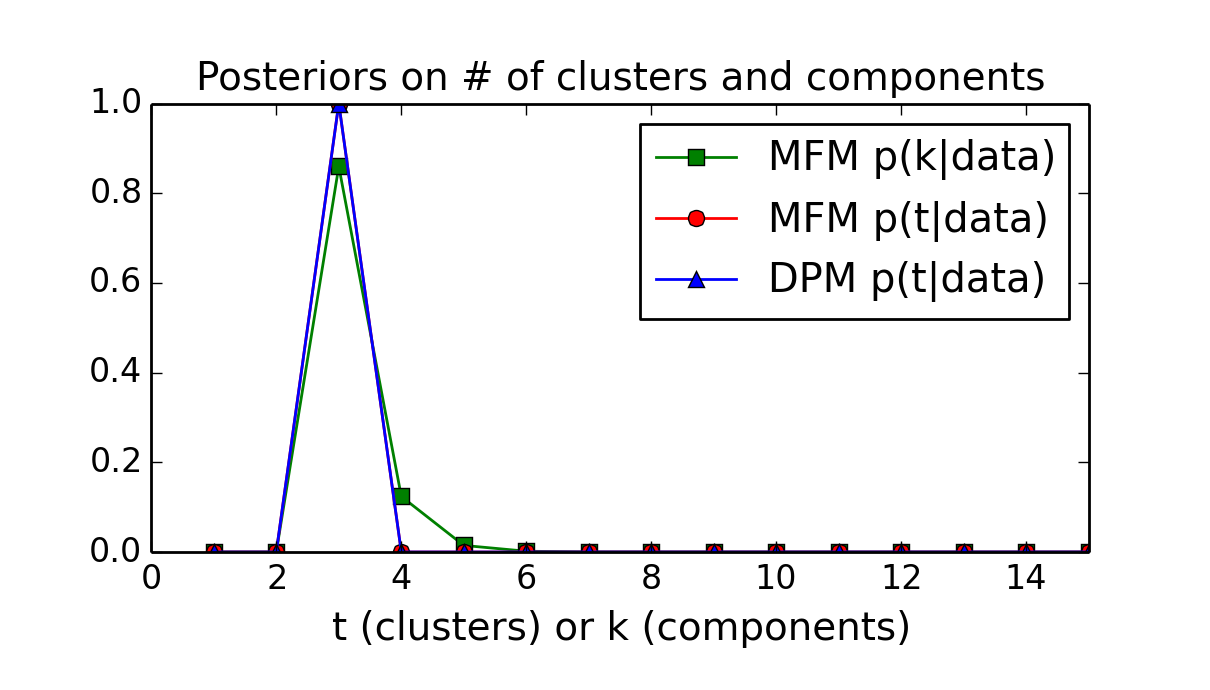}
\includegraphics[trim=7cm 3cm 5cm 0cm, clip=true, width=0.49\textwidth]{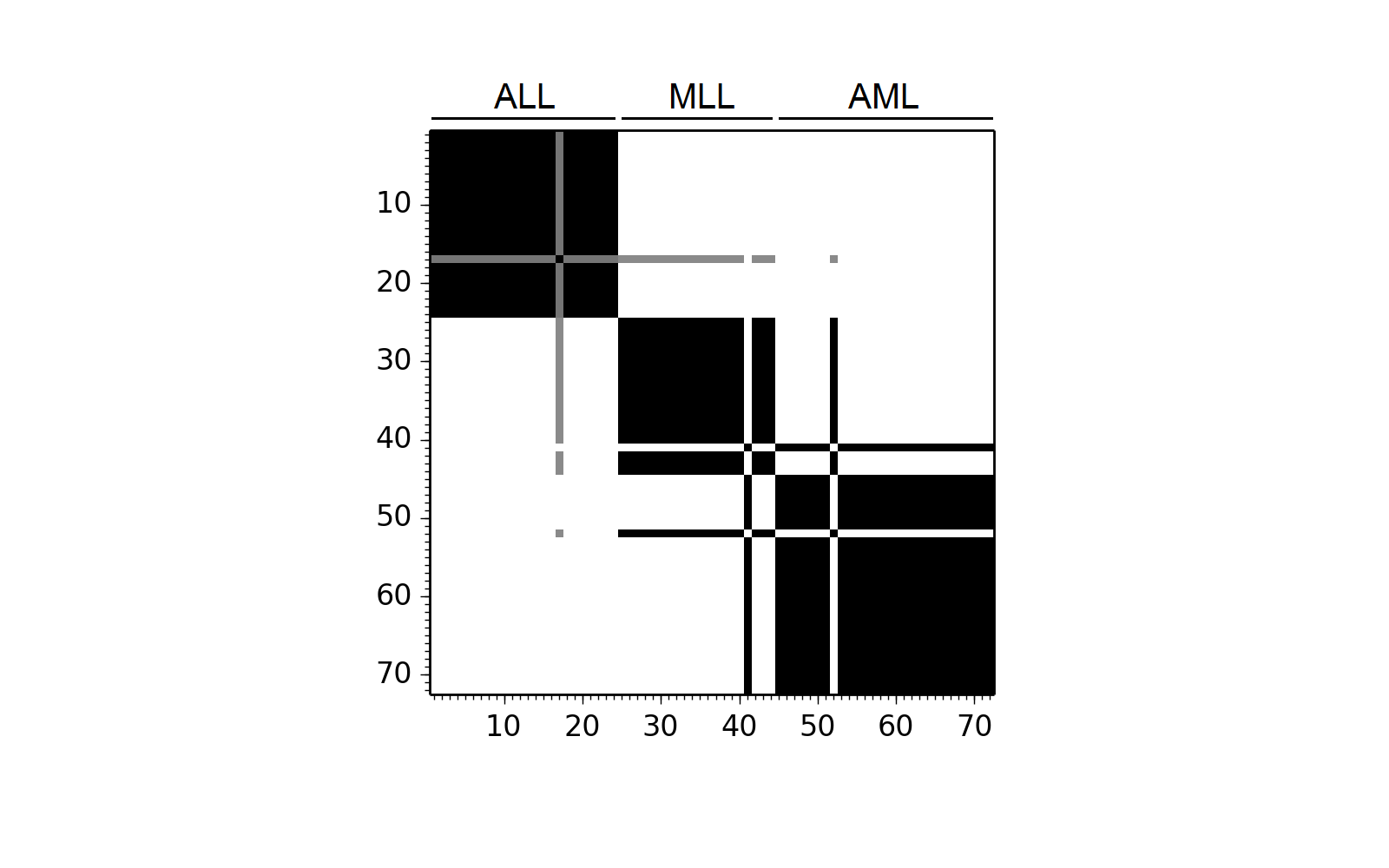}
\caption{Results on the leukemia gene expression dataset. Left: Posteriors on the number of clusters and components. Right: MFM pairwise probability matrix (the DPM matrix is the same).  See the text for discussion.}
\label{figure:gene}
\end{figure}

Figure \ref{figure:gene} also shows the MFM pairwise probability matrix, that is, the matrix in which entry $(i,j)$ is the posterior probability that data points $i$ and $j$ belong to the same cluster; in the figure, white is probability 0, black is probability 1. (The DPM matrix, not shown, is indistinguishable from the MFM matrix.)
The rows and columns of the matrix are ordered according to ground truth, such that 1--24 are ALL, 25--44 are MLL, and 45--72 are AML.
The model has clearly separated the subjects into these three groups, with a small number of exceptions: subject 41 is clustered with the AML subjects instead of MLL, subject 52 with the MLL subjects instead of AML, and subject 17 is about 50\% ALL and 50\% MLL.

\section{Discussion}
\label{section:discussion}

Due to the fact that inference for the number of components is a topic of high interest in many research communities, it seems prudent to make some cautionary remarks in this regard.
Many approaches have been proposed for estimating the number of components \citep{Henna_1985,Keribin_2000,Leroux_1992,Ishwaran_2001,James_2001,Henna_2005,Woo_2006,Woo_2007}. In theory, the MFM model provides a Bayesian approach to consistently estimating the number of components, making it a potentially attractive method of assessing the heterogeneity of the data.  However, there are several possible pitfalls to consider, some of which are more obvious than others. 

An obvious potential issue is that in many applications, the clusters which one wishes to distinguish are purely notional (for example, perhaps, clusters of images or documents), and a mixture model is used for practical purposes, rather than because the data is actually thought to arise from a mixture. Clearly, in such cases, inference for the ``true'' number of components is meaningless. On the other hand, in some applications, the data definitely comes from a mixture (for example, extracellular recordings of multiple neurons)---so there is in reality a true number of components---however, usually the form of the mixture components is far from clear. 

More subtle issues are that the posteriors on $k$ and $t$ can be
\begin{enumerate}
\item strongly affected by the base measure $H$, and
\item sensitive to misspecification of the family of component distributions $\{f_\theta\}$.
\end{enumerate}
Issue (1) can be seen, for instance, in the case of normal mixtures: it might seem desirable to choose the prior on the component means to have large variance in order to be less informative, however, this causes the posteriors on $k$ and $t$ to favor smaller values \citep{Richardson_1997,Stephens_2000,Jasra_2005}. The basic mechanism at play here is the same as in the Bartlett--Lindley paradox, and shows up in many Bayesian model selection problems.
With some care, this issue can be dealt with by varying the base measure $H$ and observing the effect on the posterior---that is, by performing a sensitivity analysis---for instance, see \cite{Richardson_1997}.

Issue (2) is more serious---in practice, we typically cannot expect our choice of $\{f_\theta:\theta\in\Theta\}$ to contain the true component densities (assuming the data is even from a mixture). When the model is misspecified in this way, the posteriors of $k$ and $t$ can be severely affected and depend strongly on $n$. For instance, if the model uses mixtures of Gaussians, and the true data distribution is not a finite mixture of Gaussians, then these posteriors can be expected to diverge to infinity as $n$ increases. Consequently, the effects of misspecification need to be carefully considered if these posteriors are to be used as measures of heterogeneity.  Steps toward addressing the issue of robustness have been taken by \cite{Woo_2006,Woo_2007} and \cite{rodriguez2014univariate}, however, this is an important problem demanding further study.

Despite these issues, sample clusterings and estimates of the number of components or clusters can provide a useful tool for exploring complex datasets, particularly in the case of high-dimensional data that cannot easily be visualized. It should always be borne in mind, though, that the results can be interpreted as being correct only to the extent that the model assumptions are correct.

\section*{Acknowledgments}

We are very grateful to Steve MacEachern for many helpful suggestions.
This work was supported in part by the National Science Foundation (NSF) grants DMS-1007593 and DMS-1309004, by the National Institute of Mental Health (NIMH) grant R01MH102840, and by the Defense Advanced Research Projects Agency (DARPA) contract FA8650-11-1-715.

\appendix

\section{Formulas for some posterior quantities}
\label{section:formulas}

Below are some details regarding computation of the posterior on $k$ and of the density estimates.

\subsubsection*{Posterior on the number of components $k$}

The posterior on $t=|\C|$ is easily estimated from posterior samples of $\C$. 
To compute the MFM posterior on $k$, note that
$$ p(k|x_{1:n}) =\sum_{t = 1}^\infty p(k|t,x_{1:n}) p(t|x_{1:n}) =\sum_{t = 1}^n p(k|t)p(t|x_{1:n}),$$
by Equation \ref{equation:XKT} and the fact that $t$ cannot exceed $n$. Using this and the formula for $p(k|t)$ given by Equation \ref{equation:pkt}, it is simple to transform our estimate of the posterior on $t$ into an estimate of the posterior on $k$. 
For the DPM, the posterior on the number of components $k$ is always trivially a point mass at infinity.

\subsubsection*{Density estimates}

Using the restaurant process (Theorem \ref{theorem:restaurant}), it is straightforward to show that if $\C$ is a partition of $[n]$ and $\phi =(\phi_c: c\in\C)$ then
\begin{align}\label{equation:posterior-predictive-C-phi}
p(x_{n +1}\mid\C,\phi,x_{1:n}) \propto 
\frac{V_{n+1}(t+1)}{V_{n +1}(t)}\gamma\, m(x_{n +1})+\sum_{c\in\C}(|c|+\gamma) f_{\phi_c}(x_{n +1})
\end{align}
where $t =|\C|$, and, using the recursion for $V_n(t)$ (Equation \ref{equation:recursion}), this is normalized when multiplied by $V_{n+1}(t)/V_n(t)$. Further,
\begin{align}\label{equation:posterior-predictive-C}
p(x_{n +1}\mid\C,x_{1:n}) \propto 
\frac{V_{n+1}(t+1)}{V_{n +1}(t)}\gamma\, m(x_{n +1})+\sum_{c\in\C}(|c|+\gamma) \frac{m(x_{c\,\cup \{n+1\}})}{m(x_c)},
\end{align}
with the same normalization constant. Therefore, when the single-cluster marginals $m(x_c)$ can be easily computed, Equation \ref{equation:posterior-predictive-C} can be used to estimate the posterior predictive density $p(x_{n +1}|x_{1:n})$ based on samples from $\C\mid x_{1:n}$. When $m(x_c)$ cannot be easily computed, Equation \ref{equation:posterior-predictive-C-phi} can be used to estimate $p(x_{n +1}|x_{1:n})$ based on samples from $\C,\phi\mid x_{1:n}$, along with samples $\theta_1,\dotsc,\theta_N\iid H$ to approximate $m(x_{n+1})\approx \frac{1}{N}\sum_{i = 1}^N f_{\theta_i}(x_{n +1})$.
(Thanks to Steve MacEachern for pointing out how to handle $m(x_{n+1})$ here.)

The posterior predictive density is, perhaps, the most natural estimate of the density.
However, following \cite{Green_2001}, a simpler way to obtain a natural estimate is by assuming that element $n+1$ is added to an existing cluster; this will be very similar to the posterior predictive density when $n$ is sufficiently large. To this end, we define 
$p_*(x_{n +1}\mid\C,\phi,x_{1:n}) = p(x_{n +1}\mid\C,\phi,x_{1:n},|\C_{n +1}|=|\C|)$, where $\C_{n+1}$ is the partition of $[n+1]$, and observe that
$$ p_*(x_{n +1}\mid\C,\phi,x_{1:n}) = \sum_{c\in\C}\frac{|c|+\gamma}{n +\gamma t}\, f_{\phi_c}(x_{n +1}) $$
where $t =|\C|$ \citep{Green_2001}. Using this, we can estimate the density by
\begin{align}\label{equation:density-estimate}
\frac{1}{N}\sum_{i = 1}^N p_*(x_{n +1}\mid\C^{(i)},\phi^{(i)},x_{1:n}),
\end{align}
where $(\C^{(1)},\phi^{(1)}),\dotsc,(\C^{(N)},\phi^{(N)})$ are samples from $\C,\phi\mid x_{1:n}$.
The corresponding expressions for the DPM are all very similar, using its restaurant process instead.
The density estimates shown in this paper are obtained using this approach. 

These formulas are conditional on additional parameters such as $\gamma$ for the MFM, and $\alpha$ for the DPM. If priors are placed on such parameters and they are sampled along with $\C$ and $\phi$ given $x_{1:n}$, then the posterior predictive density can be estimated using the same formulas as above, but also using the posterior samples of these additional parameters.

\section{Proofs}
\label{section:proofs}

\begin{proof}[Proof of Theorem \ref{theorem:EPPF}]
Letting $E_i =\{j: z_j = i\}$, and writing $\C(z)$ for the partition induced by $z=(z_1,\dotsc,z_n)$, by Dirichlet-multinomial conjugacy we have
\begin{align*}
p(z|k) &=\int p(z|\pi) p(\pi|k) d\pi 
 = \frac{\Gamma(k\gamma)}{\Gamma(\gamma)^k}\frac{\prod_{i = 1}^k \Gamma(|E_i| +\gamma)}{\Gamma(n + k\gamma)}
=\frac{1}{(k\gamma)^{(n)}}\prod_{c\in\C(z)} \gamma^{(|c|)},
\end{align*}
for $z\in[k]^n$, provided that $p_K(k)>0$.
It follows that for any partition $\C$ of $[n]$,
\begin{align}
p(\C|k) &=\sum_{z\in[k]^n\,:\,\C(z) =\C} p(z|k) \notag\\
&=\#\Big\{z\in[k]^n:\C(z) =\C\Big\}\,\frac{1}{(\gamma k)^{(n)}}\prod_{c\in\C}\gamma^{(|c|)}\notag\\
&=\frac{k_{(t)}}{(\gamma k)^{(n)}}\prod_{c\in\C}\gamma^{(|c|)},\label{equation:pCk}
\end{align}
where $t =|\C|$, since $\#\big\{z\in[k]^n:\C(z) =\C\big\}={k\choose t} t! = k_{(t)}$.
Finally,
\begin{align*}
p(\C)&=\sum_{k = 1}^\infty p(\C|k) p_K(k)
=\Big(\prod_{c\in\C}\gamma^{(|c|)}\Big)\sum_{k = 1}^\infty\frac{k_{(t)}}{(\gamma k)^{(n)}}\,p_K(k) 
= V_n(t)\prod_{c\in\C}\gamma^{(|c|)},
\end{align*}
with $V_n(t)$ as in Equation \ref{equation:v}.
\end{proof}

\begin{proof}[Proof of Equation \ref{equation:model-C}]
Theorem \ref{theorem:EPPF} shows that the distribution of $\C$ is as shown.
Next, note that instead of sampling only $\theta_1,\dotsc,\theta_k\iid H$ given $K = k$, we could simply sample $\theta_1,\theta_2,\dotsc\iid H$ independently of $K$, and the distribution of $X_{1:n}$ would be the same.
Now, $Z_{1:n}$ determines which subset of the i.i.d.\ variables $\theta_1,\theta_2,\dotsc$ will actually be used, and the indices of this subset are independent of $\theta_1,\theta_2,\dotsc$; hence, denoting these random indices $I_1<\cdots<I_T$, we have that $\theta_{I_1},\dotsc,\theta_{I_T}|Z_{1:n}$ are i.i.d.\ from $H$. 
For  $c\in\C$, let $\phi_c =\theta_{I_i}$ where $i$ is such that $c =\{j: z_j = I_i\}$. This completes the proof.
\end{proof}

\begin{proof}[Proof of the properties in Section \ref{section:basic}]
\label{section:basic-derivation}

Abbreviate $x = x_{1:n}$, $z = z_{1:n}$, and $\theta =\theta_{1:k}$, and assume $p(z,k)>0$.
Letting $E_i =\{j: z_j = i\}$, we have $p(x|\theta,z,k) = \prod_{i = 1}^k\prod_{j\in E_i} f_{\theta_i}(x_j)$
and
\begin{align*}
p(x|z,k) &= \int_{\Theta^k} p(x|\theta,z,k) p(d\theta|k)
= \prod_{i = 1}^k \int_\Theta \Big[\prod_{j\in E_i} f_{\theta_i}(x_j)\Big] H(d\theta_i)\\
&= \prod_{i = 1}^k m(x_{E_i})=\prod_{c\in\C(z)}m(x_c).
\end{align*}
Since this last expression depends only on $z,k$ through $\C =\C(z)$, we have
$p(x|\C) = \prod_{c\in\C} m(x_c)$,
establishing Equation \ref{equation:marginal-product}.
Next, recall that
$p(\C|k) =\frac{k_{(t)}}{(\gamma k)^{(n)}}\prod_{c\in\C}\gamma^{(|c|)}$ (where $t =|\C|$)
from Equation \ref{equation:pCk}, and thus
$$ p(t|k)=\sum_{\C:|\C|=t} p(\C|k) =\frac{k_{(t)}}{(\gamma k)^{(n)}}\sum_{\C:|\C|= t}\prod_{c\in\C}\gamma^{(|c|)},$$
(where the sum is over partitions $\C$ of $[n]$ such that $|\C|= t$) establishing Equation \ref{equation:ptk}.
Equation \ref{equation:pkt} follows, since
$$ p(k|t)\propto p(t|k) p(k) \propto \frac{k_{(t)}}{(\gamma k)^{(n)}}\,p_K(k), $$
(provided $p(t)>0$) and the normalizing constant is precisely $V_n(t)$.
To see that $\C\perp K\mid T$ (Equation \ref{equation:CKT}), note that if $t =|\C|$ then
$$ p(\C|t,k) =\frac{p(\C,t|k)}{p(t|k)} =\frac{p(\C|k)}{p(t|k)},$$
(provided $p(t,k)>0$) and due to the form of $p(\C|k)$ and $p(t|k)$ just above, this quantity does not depend on $k$; hence, $p(\C|t,k) = p(\C|t)$. To see that $X \perp K\mid T$ (Equation \ref{equation:XKT}), note that $X\perp K\mid\C$; using this in addition to $\C\perp K\mid T$, we have
$$ p(x|t,k)  =\sum_{\C:|\C|= t} p(x|\C,t,k)p(\C|t,k) = \sum_{\C:|\C|= t} p(x|\C,t) p(\C|t) = p(x|t). $$
\end{proof}

\begin{proof}[Proof of Theorem \ref{theorem:restaurant}]
Let $\C_\infty$ be the random partition of $\Z_{>0}$ as in Section \ref{section:self-consistent}, and for $n\in\{1,2,\dotsc\}$, let $\C_n$ be the partition of $[n]$ induced by $\C_\infty$. Then
$$ p(\C_n|\C_{n-1},\dotsc,\C_1) = p(\C_n|\C_{n-1}) \propto q_n(\C_n)\,I(\C_n\setminus n = \C_{n-1}), $$
where $\C\setminus n$ denotes $\C$ with element $n$ removed, and $I(\cdot)$ is the indicator function ($I(E) = 1$ if $E$ is true, and $I(E) = 0$ otherwise).
Recalling that $q_n(\C_n) = V_n(|\C_n|)\prod_{c\in\C_n}\gamma^{(|c|)}$ (Equation \ref{equation:EPPF}), we have, letting $t =|\C_{n-1}|$,
$$ p(\C_n|\C_{n-1}) \propto \branch{V_n(t+1)\gamma}{\mbox{$n$ is a singleton in $\C_n$, i.e., $\{n\}\in\C_n$}}
                                   {V_n(t)(\gamma +|c|)}{c\in\C_{n-1}\mbox{ and } c\cup\{n\}\in\C_n,}$$
for $\C_n$ such that $\C_n\setminus n =\C_{n-1}$ (and $p(\C_n|\C_{n-1})=0$ otherwise).  
With probability $1$, $q_{n - 1}(\C_{n-1})>0$, thus $V_{n-1}(t)>0$ and hence also $V_n(t)>0$, so we can divide through by $V_n(t)$ to get the result.
\end{proof}

\begin{proof}[Proof of Theorem \ref{theorem:species-predictive}]
Let $G\sim\M(p_K,\gamma,H)$ and let $\beta_1,\dotsc,\beta_n\iid G$, given $G$. Then the joint distribution of $(\beta_1,\dotsc,\beta_n)$ (with $G$ marginalized out) is the same as $(\theta_{Z_1},\dotsc,\theta_{Z_n})$ in the original model (Equation \ref{equation:model}). Let $\C_n$ denote the partition induced by $Z_1,\dotsc,Z_n$ as usual, and for $c\in\C_n$, define $\phi_c=\theta_I$ where $I$ is such that $c=\{j : Z_j = I\}$; then, as in the proof of Equation \ref{equation:model-C}, $(\phi_c: c\in\C_n)$ are i.i.d.\ from $H$, given $\C_n$.

Therefore, we have the following equivalent construction for $(\beta_1,\dotsc,\beta_n)$:
\begin{align*}
& \C_n\sim q_n, \mbox{ with $q_n$ as in Section \ref{section:self-consistent}}\\
& \phi_c\iid H \mbox{ for $c\in\C_n$, given $\C_n$} \\
& \beta_j =\phi_c \mbox{ for $j\in c$, $c\in\C_n$, given $\C_n,\phi$.}
\end{align*}
Due to the self-consistency property of $q_1,q_2,\dotsc$ (Proposition \ref{proposition:self-consistent}), we can sample $\C_n,(\phi_c: c\in\C_n),\beta_{1:n}$ sequentially for $n = 1,2,\dotsc$ by sampling from the restaurant process for $\C_n|\C_{n-1}$, sampling $\phi_{\{n\}}$ from $H$ if $n$ is placed in a cluster by itself (or setting $\phi_{c\,\cup\{n\}}=\phi_c$ if $n$ is added to $c\in\C_{n-1}$), and setting $\beta_n$ accordingly. 

In particular, if the base measure $H$ is continuous, then the $\phi$'s are distinct with probability $1$, so conditioning on $\beta_{1:n-1}$ is the same as conditioning on $\C_{n-1},(\phi_c:c\in\C_{n-1}),\beta_{1:n-1}$, and hence we can sample $\beta_n|\beta_{1:n-1}$ in the same way as was just described. In view of the form of the restaurant process (Theorem \ref{theorem:restaurant}), the result follows.
\end{proof}




We use the following elementary result in the proof of Theorem \ref{theorem:v-asymptotics}; it is a special case of the dominated convergence theorem.

\begin{proposition}\label{proposition:sums}
For $j = 1,2,\dotsc$, let $a_{1j}\geq a_{2j}\geq\cdots\geq 0$ such that $a_{ij}\to 0$ as $i\to\infty$. If $\sum_{j = 1}^\infty a_{1j}<\infty$ then $\sum_{j = 1}^\infty a_{ij}\to 0$ as $i\to\infty$.
\end{proposition}

\begin{proof}[Proof of Theorem \ref{theorem:v-asymptotics}]
For any $x>0$, writing $x^{(n)}/n! =\Gamma(x + n)/(n!\,\Gamma(x))$ and using Stirling's approximation, we have
$$\frac{x^{(n)}}{n!}\sim\frac{n^{x - 1}}{\Gamma(x)}$$
as $n\to\infty$.
Therefore, the $k = t$ term of $V_n(t)$ is
$$\frac{t_{(t)}}{(\gamma t)^{(n)}}\,p_K(t) \sim \frac{t!}{n!}\,\frac{\Gamma(\gamma t)}{n^{\gamma t-1}}\,p_K(t). $$
The first $t - 1$ terms of $V_n(t)$ are $0$, so to prove the result, we need to show that the rest of the series, divided by the $k = t$ term, goes to $0$. (Recall that we have assumed $p_K(t)>0$.) To this end, let
$$ b_{nk} = (\gamma t)^{(n)}\frac{k_{(t)}}{(\gamma k)^{(n)}}\,p_K(k). $$
We must show that $\sum_{k = t +1}^\infty b_{nk}\to 0$ as $n\to\infty$. We apply Proposition \ref{proposition:sums} with $a_{ij} = b_{t + i,t + j}$. For any $k>t$, $b_{1k}\geq b_{2k}\geq\cdots\geq 0$. Further, for any $k>t$,
$$\frac{(\gamma t)^{(n)}}{(\gamma k)^{(n)}}\sim\frac{n^{\gamma t - 1}}{\Gamma(\gamma t)}\,\frac{\Gamma(\gamma  k)}{n^{\gamma k - 1}}\longrightarrow 0 $$
as $n\to\infty$, hence, $b_{nk}\to 0$ as $n\to\infty$ (for any $k>t$). Finally, observe that $\sum_{k = t +1}^\infty b_{nk}\leq(\gamma t)^{(n)} V_n(t)<\infty$ for any $n\geq t$. Therefore, by Proposition \ref{proposition:sums}, $\sum_{k = t +1}^\infty b_{nk}\to 0$ as $n\to\infty$. This proves the result.
\end{proof}

\begin{proof}[Proof of Theorem \ref{theorem:clusters-components-relationship}]
For any $t\in\{1,\dotsc,k\}$,
\begin{align}\label{equation:pkt-limit}
p_n(K = t\mid T = t)= \frac{1}{V_n(t)}\frac{t_{(t)}}{(\gamma t)^{(n)}}\,p_K(t)\longrightarrow 1
\end{align}
as $n\to \infty$ (where $p_n$ denotes the MFM distribution with $n$ samples), by Equation \ref{equation:pkt} and Theorem \ref{theorem:v-asymptotics}. For any $n\geq k$,
\begin{align*}
p(K = k\mid x_{1:n}) =\sum_{t = 1}^k p(K = k\mid T = t,x_{1:n})\,p(T = t \mid x_{1:n}),
\end{align*}
and note that by Equations \ref{equation:XKT} and \ref{equation:pkt-limit}, $p(K = k\mid T = t,x_{1:n}) = p_n(K = k\mid T = t)\longrightarrow I(k = t)$ for $t\leq k$. The result follows.
\end{proof}


\bibliographystyle{abbrvnatcap}
\bibliography{paper}

\end{document}